\documentclass[11pt, a4paper]{article}
\usepackage{amsmath}
\usepackage[english]{babel}
\usepackage{amsmath}
\usepackage{amsfonts}
\usepackage{mathrsfs}
\usepackage{amssymb}
\usepackage{graphicx}
\usepackage{amsthm,amssymb}
\usepackage{epsfig}
\usepackage{float}
\usepackage{color}
\usepackage{pdflscape}
\usepackage[top=1in, bottom=1in, left=1in, right=1in]{geometry}
\usepackage{yfonts}
\usepackage{algorithm}
\usepackage{algpseudocode}
\usepackage{longtable}
\usepackage{braket}
\usepackage{amsfonts}
\usepackage{subfigure}
\numberwithin{figure}{section}
\newtheorem{theorem}{Theorem}[section]
\newtheorem{lemma}{Lemma}[section]
\newtheorem{corollary}{Corollary}[section]
\newtheorem{proposition}{Proposition}[section]
\newtheorem{example}{Example}[section]
\newtheorem{remark}{Remark}[section]
\newtheorem{definition}{Definition}[section]

\date{}
\usepackage{tikz}
\usetikzlibrary{matrix,arrows,shapes.geometric}
\usepackage{longtable}
\usepackage{geometry,mathtools}
\begin{document}
	\title{
  Encoding and Construction of Quantum Codes from $(\gamma,\Delta)$-cyclic Codes over a Class of Non-chain Rings
 }
	\author{\small{Om Prakash$^{* 1}$, Shikha Patel$^{2}$ and Habibul Islam$^{3}$} \\
		\small{$^{1,2}$Department of Mathematics} \\ \small{Indian Institute of Technology Patna} \\ \small{Bihta, Patna - 801 106, India} \\
 \small{$^{3}$ Department of Mathematics } \\ \small{Indian Institute of Information Technology Bhopal, India} \\
  \small{om@iitp.ac.in (*corresponding author), shikha\_1821ma05@iitp.ac.in, habibul.islam@iiitbhopal.ac.in}}

\maketitle

\begin{abstract}
\noindent Let $\mathbb{F}_q$ be a finite field of $q=p^m$ elements where $p$ is a prime and $m$ is a positive integer. This paper considers $(\gamma,\Delta)$-cyclic codes over a class of finite non-chain commutative rings $\mathscr{R}_{q,s}=\mathbb{F}_q[v_1,v_2,\dots,v_s]/\langle v_i-v_i^2,v_iv_j=v_jv_i=0\rangle$ where $\gamma$ is an automorphism of $\mathscr{R}_{q,s}$, $\Delta$ is a $\gamma$-derivation of $\mathscr{R}_{q,s}$ and $1\leq i\neq j\leq s$ for a positive integer $s$. Here, we show that a $(\gamma,\Delta)$-cyclic code of length $n$ over $\mathscr{R}_{q,s}$ is the direct sum of $(\theta,\Im)$-cyclic codes of length $n$ over $\mathbb{F}_q$, where $\theta$ is an automorphism of $\mathbb{F}_q$ and $\Im$ is a $\theta$-derivation of $\mathbb{F}_q$. Further, necessary and sufficient conditions for both $(\gamma,\Delta)$-cyclic and $(\theta,\Im)$-cyclic codes to contain their Euclidean duals are established. Then, we obtain many quantum codes by applying the dual containing criterion on the Gray images of these codes. These codes have better parameters than those available in the literature. Finally, the encoding and error-correction procedures for our proposed quantum codes are discussed.\\

CSS construction, quantum codes.
\end{abstract}

%\linenumbers

\section{Introduction}
After the pioneering work of Hammons et al. \cite{Hammons} in 1994, codes over finite rings attracted many researchers for better error-correcting codes. Later, several important research has been carried out over finite rings and explored plenty of suitable parameters; we refer \cite{Alahmadi21,Ashraf19,Dertli15,Gao19,Ma18}.
Nevertheless, all of these works have been considered over finite commutative rings. Hence, it is natural to look at these works over the noncommutative ring to obtain codes with better parameters. Towards this, in 2007, Boucher et al. \cite{Boucher1} introduced skew cyclic codes as a generalized class of cyclic codes using a non-trivial automorphism $\theta$ on a finite field $\mathbb{F}_q$. They proved that the noncommutative rings (skew polynomial rings) are worthy alphabets for producing new parameters. In addition, they have provided a few codes with better parameters that were not known earlier over finite commutative rings. The factorization of the polynomial $x^n-1$ plays an important role in the characterization of cyclic codes of length $n$ and more factorization leads to the case of getting many new codes with better parameters. Therefore, skew cyclic codes that generalize cyclic codes in a noncommutative setup attract many researchers. During 2010-2012,  Abualrub et al. \cite{Abualrub} and Bhaintwal \cite{Bhaintwal} introduced and developed some interesting results on skew-quasi cyclic codes. From an application point of view, recently, many authors have shown that skew cyclic codes are one of the important resources for producing new quantum codes along with classical codes \cite{Boucher1,Boucher11,Boucher2,Gursoy,Jitman,Siap}. \\
However, all of the above works have been carried out on skew polynomial rings of automorphism type. Only a few works are available in the literature with both automorphisms and derivations. In \cite{Boucher,MH,P,Luis}, the authors generalized the notion of codes over skew polynomial rings with non-trivial automorphism $\theta$ and $\theta$- derivation $\Im$ under the usual addition of polynomials and a specific polynomials multiplication involving $\theta$ and $\Im$. For the noncommutative ring $\mathbb{F}_q[x,\theta;\Im]$ where $\theta$ is the Frobenius automorphism $a\mapsto a^p$, $p$ is the characteristic of $\mathbb{F}_q$, the authors \cite{Boucher,Luis} defined the inner $\theta$-derivation $\Im$ induced by $\beta\in \mathbb{F}_q^*$ of the form $a\mapsto  \beta(\theta(a)-a)$. Further, Boulagouaz and Leroy \cite{MH} studied $(\sigma,\delta)$-codes with  $\sigma$-derivation induced by the ring element. Recently, Sharma and Bhaintwal \cite{Sharma} have studied skew cyclic codes over $\mathbb{Z}_4+u\mathbb{Z}_4,$ $u^2=1$ with both automorphism and inner derivation. In 2021, Ma et al. \cite{Li} studied  $(\sigma,\delta)$-skew quasi-cyclic codes over the ring $\mathbb{Z}_4+u\mathbb{Z}_4,$ $u^2=1$. Further, in 2021, Patel and Prakash \cite{Patel21} studied $(\theta,\delta_\theta)$-cyclic codes over the ring $\mathbb{F}_q[u,v]/\langle u^2-u, v^2-v, uv-vu \rangle$ via the decomposition method over $\mathbb{F}_q$. Here, we extend our previous work \cite{Patel21} to a more general structure and propose a fruitful application of $(\gamma,\Delta)$-cyclic codes in the context of quantum code construction. As per our survey, it is worth mentioning that this is the first article proposing an application of $(\gamma,\Delta)$-cyclic codes into quantum codes.\\
Quantum error-correcting codes play a significant role in protecting information against disturbances such as decoherence occurring in the channel. In this connection, in 1995, Shor \cite{Shor} discovered the first quantum code. After that, Calderbank et al. \cite{Calderbank98} provided a method to obtain quantum codes from classical codes. This technique became very popular among researchers and is known as the CSS (Calderbank-Shor-Steane) construction. Presently, quantum codes and their implementation from classical codes have gained significant attention. As a consequence, many quantum codes with better parameters have been constructed from different families of linear codes such as cyclic, skew cyclic, skew constacyclic codes, etc., see \cite{Alahmadi21,Ashraf19,Dertli15,Gao19,Li20,Ma18,Ozen19,Prakash}. However, the search for new methods on different structures is still ongoing by which one can construct quantum codes efficiently with suitable parameters. Since getting new quantum codes proportionally depends on the abundance of factors of $x^n-1$, many authors have been exploring quantum codes in the setting of the skew polynomial ring with automorphism where $x^n-1$ indeed possesses more factorization than the commutative case. Thus, in this work, we extend all these previous works in a new direction by considering skew polynomial rings with non-trivial automorphisms and nonzero derivations. Here, we use different derivations for the same Frobenius automorphism having the form $a\mapsto  \beta(\theta(a)-a)$ for all $\beta\in \mathbb{F}_q^*$.\\
The rest of the paper is structured as follows: In Section $2$, we present some basic results and notations that will be useful for later sections. In Section $3$, we discuss $(\theta,\Im)$-cyclic codes over $\mathbb{F}_q$ and derive a necessary and sufficient condition to contain their duals over $\mathbb{F}_q$. Further, Section $4$ includes the results on $(\gamma,\Delta)$-cyclic codes over $\mathscr{R}_{q,s}$ and dual-containing property for these codes as well. Section $5$ describes the applications of our obtained results by providing many new quantum codes with superior parameters. Finally, Section $6$ concludes our work.
\section{Preliminaries}
In this Section, we provide some preliminary results, definitions and notations which are used throughout this paper. We consider a finite non-chain ring $\mathscr{R}_{q,s}:=\mathbb{F}_q[v_1,v_2,\dots,v_s]/\langle v_i-v_i^2,v_iv_j=v_jv_i=0\rangle$ where $1\leq i\neq j\leq s$ and $s$ is a positive integer. This $\mathscr{R}_{q,s}$ is a class of finite commutative ring with unity for different values of $q$ and $s$. Further, $\mathscr{R}_{q,s}$ can also be represented in the form of $\mathscr{R}_{q,s}=\mathbb{F}_q+v_1\mathbb{F}_q+\cdots+v_s\mathbb{F}_q$ with $v_i-v_i^2,v_iv_j=v_jv_i=0$. Moreover, $\mathscr{R}_{q,s}$ is a non-chain semi-local Frobenius ring having $s + 1$ maximal ideals. For $s=2$, there are three maximal ideals $\langle v_1+v_2\rangle$, $\langle 1-v_1\rangle$ and $\langle 1-v_2\rangle$ in $\mathscr{R}_{q,2}$, refer \cite{Islam20b}. Consider
$$\zeta_0=\displaystyle\prod_{i=1}^{s}(1-v_i), ~~\text{and}~~ \zeta_j=v_j, ~~1\leq j\leq s.$$ It is easy to verify that $\sum_{i=0}^{s} \zeta_i=1$ and
\begin{equation*}
	\zeta_i\zeta_j=
	\begin{cases}
		\zeta_i, & \mbox if~ i= j\\
		0, & \mbox if ~i\neq j
	\end{cases}.
\end{equation*}
Hence, by Chinese Remainder Theorem, $\mathscr{R}_{q,s}=\zeta_0\mathscr{R}_{q,s}\oplus\zeta_1\mathscr{R}_{q,s}\oplus\cdots\oplus\zeta_s\mathscr{R}_{q,s}=\zeta_0\mathbb{F}_q\oplus\zeta_1\mathbb{F}_q\oplus\cdots\oplus\zeta_s\mathbb{F}_q$. Thus, we conclude that any element $t\in\mathscr{R}_{q,s}$ can be uniquely written as $t=\zeta_0t_0+\zeta_1t_1+\cdots+\zeta_st_s$, where $t_i\in \mathbb{F}_q$. Also, $t$ is a unit in $\mathscr{R}_{q,s}$ if and only if $t_i\in \mathbb{F}_q^*$ for all $i$. \\
Recall that a non-empty subset $\mathcal{C}$ of $\mathscr{R}_{q,s}^n$ is said to be a linear code of length $n$ over $\mathscr{R}_{q,s}$ if it is an $\mathscr{R}_{q,s}$-submodule of $\mathscr{R}_{q,s}^n$ and the elements of $\mathcal{C}$ are called codewords. The Hamming weight $w_H(c)$ of a codeword $c=(c_0,c_1,\dots,c_{n-1})\in \mathcal{C} $ is the number of nonzero coordinates in $c$. The Hamming distance between any two codewords $c$ and $c'$ of $\mathcal{C}$ is defined as $d_H(c,c')=w_H(c-c')$ and the Hamming distance of a linear code $\mathcal{C}$ is defined as $d_H(\mathcal{C})= \min \{ d_H(x,y)~|~ x, ~y\in \mathcal{C}, x \neq y \}.$ The Euclidean inner product of $c$ and $c'$ in $\mathcal{R}^n$ is defined by $c\cdot c' = \sum_{i=0}^{n-1}c_ic'_i$ where $c=(c_0,c_1,\dots,c_{n-1})$ and $c'=(c'_0,c'_1,\dots,c'_{n-1})$	are codewords in $\mathcal{C}$. The dual code of $\mathcal{C}$ is defined by  $\mathcal{C}^\perp= \{c\in \mathscr{R}_{q,s}^n~|~ c\cdot c' = 0, ~\text{for all} ~c'\in \mathcal{C} \}$. Also, a linear code $\mathcal{C}$ is self-orthogonal if $\mathcal{C}\subseteq \mathcal{C}^\perp$ and self-dual if $\mathcal{C}=\mathcal{C}^\perp$. Further, let $c=(c_0,c_1,\dots,c_{n-1})\in\mathcal{C}\subseteq  \mathbb{F}_q^n$. If $\mathcal{C}$ is an $[n,k,d]$ linear code, then from the Singleton bound, its minimum distance is bounded above by $d \leq n - k + 1$, where $d$ is the minimum distance, $k$ is the dimension, and $n$ is the length of the code. A code achieving the mentioned bound is called maximum-distance-separable (MDS). If the minimum distance of the code is one unit less than the MDS, then the code is called almost MDS. A linear code is said to be optimal if it has the highest possible minimum distance for a given length and dimension.
\begin{definition}
	Let $\mathscr{R}_{q,s}$ be a finite ring and $\gamma$ be an automorphism of $\mathscr{R}_{q,s}$. Then a map $\Delta : \mathscr{R}_{q,s} \rightarrow \mathscr{R}_{q,s}$ is said to be a $\gamma$-derivation of $\mathscr{R}_{q,s}$ if
	\begin{enumerate}
		\item $\Delta(x+y)=\Delta(x)+\Delta(y)$;
		\item $\Delta(xy)=\Delta(x)y+\gamma(x)\Delta(y)$
	\end{enumerate} for all $x,y\in \mathscr{R}_{q,s}.$
\end{definition}
Let us consider an automorphism $\theta:\mathbb{F}_q\rightarrow \mathbb{F}_q $ defined by $\theta(a)=a^q$, for all $a\in \mathbb{F}_q$ and a $\theta$-derivation $\Im:\mathbb{F}_q\rightarrow \mathbb{F}_q $ defined by $\Im(a)=\beta(\theta(a)-a)$, for all $a\in \mathbb{F}_q$ and $\beta\in \mathbb{F}_q^*$. Now, we extend the above maps over $\mathscr{R}_{q,s}$ and define the skew polynomial ring with both automorphism and derivation over $\mathscr{R}_{q,s}$. Let $Aut(\mathscr{R}_{q,s})$ be the set of all automorphism of $\mathscr{R}_{q,s}$ and $\gamma\in Aut(\mathscr{R}_{q,s})$. We consider the set
$$\mathscr{R}_{q,s}[x;\gamma, \Delta]=\{b_lx^l+\cdots+b_1x+b_0~|~ b_i\in \mathscr{R}~ \text{and}~ l\in \mathbb{N}\},$$ where $\Delta$ is a $\gamma$-derivation of $\mathscr{R}_{q,s}$. Then $\mathscr{R}_{q,s}[x;\gamma, \Delta]$ is a noncommutative ring unless $\gamma$ is the identity under the usual addition of polynomials and multiplication is defined with respect to $xb=\gamma(b)x+\Delta(b)$ for $b\in \mathscr{R}_{q,s}$, known as a skew polynomial ring.
\begin{definition}
	An element $f(x)\in\mathscr{R}_{q,s}[x;\gamma, \Delta]$ is said to be a central element of $\mathscr{R}_{q,s}[x;\gamma, \Delta]$ if $f(x)b(x)=b(x)f(x),$ for all $b(x)\in \mathscr{R}_{q,s}[x;\gamma, \Delta]$.
\end{definition}
\begin{definition} \cite{N,Leroy}
	A pseudo-linear transformation $T_{\gamma,\Delta}: \mathscr{R}_{q,s}^n \rightarrow \mathscr{R}_{q,s}^n$ is an additive map defined by
	\begin{equation}\label{eq1}
		T_{\gamma,\Delta}(v)= \gamma(v)M + \Delta(v),
	\end{equation}
	where $v= (v_1,v_2,\dots,v_n)\in \mathscr{R}_{q,s}^n$, $\gamma(v)= (\gamma(v_1),\gamma(v_2),\dots,\gamma(v_n))\in \mathscr{R}_{q,s}^n$, $M$ is a matrix of order $n\times n$ over $\mathscr{R}_{q,s}$ and $\Delta(v)= (\Delta(v_1),\Delta(v_2),\dots,\Delta(v_n))\in \mathscr{R}_{q,s}^n$. If $\Delta=0$, then $T_{\gamma}$ is known as \textit{semi-linear transformation}.
\end{definition}
\begin{definition}
	\begin{enumerate}
		\item A code $\mathcal{C}$ of length $n$ over $\mathscr{R}_{q,s}$ is said to be a $(\gamma, \Delta)$-linear code if it is a left $\mathscr{R}_{q,s}[x;\gamma, \Delta]$-submodule of $\frac{\mathscr{R}_{q,s}[x;\gamma, \Delta]}{\langle x^n-1 \rangle}$. Moreover, if $x^n-1$ is a central element of $\mathscr{R}_{q,s}[x;\gamma, \Delta]$, then $\mathcal{C}$ is a central $(\gamma, \Delta)$-linear code.
		\item A code $\mathcal{C}$ of length $n$ over $\mathscr{R}_{q,s}$ is said to be a $(\gamma, \Delta)$-cyclic code if
		\begin{itemize}
			\item $\mathcal{C}$ is a $(\gamma, \Delta)$-linear code;
			\item $T_{\gamma,\Delta}(\mathcal{C})\subseteq \mathcal{C}$, where $T_{\gamma,\Delta}$ is as defined in Equation (\ref{eq1}) and $M$ is defined as
			$$M=\begin{pmatrix}
				0 & 1  & \dots & 0\\
				\vdots & \vdots & \ddots & \vdots\\
				0 & 0  & \dots & 1\\
				1 & 0 & \dots & 0\\
			\end{pmatrix}.$$
		\end{itemize}
	\end{enumerate}
\end{definition}
\begin{remark}\cite[Exercise~20]{KR}\label{re1}
	Let $\mathscr{R}_{q,s}[x;\gamma, \Delta]$ be a skew polynomial ring, $r\in \mathscr{R}_{q,s}$ and $n\in \mathbb{N}$. Then
	$$x^nr=\gamma^n(r)x^n+a_{n-1}x^{n-1}+\cdots+a_1x+\Delta^n(r),$$ for some $a_{n-1},\dots, a_1\in\mathscr{R}_{q,s}$.
\end{remark}
To find generator polynomials of $( \gamma, \Delta)$-cyclic codes over $\mathscr{R}_{q,s},$, first we derive the right division algorithm in $\mathscr{R}_{q,s}[x;\gamma, \Delta]$.
\begin{theorem}( The Right Division Algorithm)\label{th algo}
	Let $f(x),~g(x)\in \mathscr{R}_{q,s}[x;\gamma, \Delta]$ such that the leading coefficient of $g(x)$ be a unit. Then there exist $q(x),~r(x)\in \mathscr{R}_{q,s}[x;\gamma, \Delta]$ such that
	$$f(x)=q(x)g(x)+r(x),$$ where $r(x)=0$ or $\deg~r(x)<\deg~g(x)$.
\end{theorem}
\begin{proof}
	
	If $f(x)=0$, then the result follows by taking $q(x),~r(x)=0$. If $\deg~f(x)<\deg~g(x)$, then we take $q(x)=0$ and $r(x)=f(x)$.
	Furthermore, for $\deg f(x)\geq \deg g(x)$, we prove it by induction on $\deg f(x)$.\\
	It can be easily seen that the result is true for $\deg f(x)=0$. Now, suppose the result is true for all polynomials of degree less than $\deg~f(x)$. Let $f(x)=f_0+f_1x+\cdots+f_sx^s$ and $g(x)=g_0+g_1x+\cdots+g_tx^t$ be two polynomials in $\mathscr{R}_{q,s}[x;\gamma, \Delta]$ such that $f_s\neq0$ and $g_t$ is a unit. Consider a polynomial
	$$h(x)=f(x)-f_s\gamma^{s-t}(g_t^{-1})x^{s-t}g(x).$$ From Remark \ref{re1}, $h(x)$ can be written as
    \allowdisplaybreaks
	\begin{align*}
		h(x)=&f(x)-f_s\gamma^{s-t}(g_t^{-1})x^{s-t}g(x)
		=f(x)-f_s\gamma^{s-t}(g_t^{-1})x^{s-t}(g_0+g_1x+\cdots+g_tx^t)\\
		=&f(x)-f_s\gamma^{s-t}(g_t^{-1})x^{s-t}g_0-f_s\gamma^{s-t}(g_t^{-1})x^{s-t}g_1x-\cdots-f_s\gamma^{s-t}(g_t^{-1})x^{s-t}g_tx^t\\
		=&f(x)-f_s\gamma^{s-t}(g_t^{-1})(\gamma^{s-t}(g_0)x^{s-t}+a_{s-t-1}x^{s-t-1}+\cdots+a_1x+\Delta^{s-t}(g_0))\\&
		-\cdots-f_s\gamma^{s-t}(g_t^{-1})(\gamma^{s-t}(g_t)x^{s-t}+b_{s-t-1}x^{s-t-1}+\cdots+b_1x+\Delta^{s-t}(g_t))x^t\\
		=&f(x)-f_s\gamma^{s-t}(g_t^{-1})(\gamma^{s-t}(g_0)x^{s-t}+a_{s-t-1}x^{s-t-1}+\cdots+a_1x+\Delta^{s-t}(g_0))\\&
		 -\cdots-f_s\gamma^{s-t}(g_t^{-1})\gamma^{s-t}(g_t)x^{s-t}x^t-f_s\gamma^{s-t}(g_t^{-1})b_{s-t-1}x^{s-t-1}x^t-\cdots-\\&f_s\gamma^{s-t}(g_t^{-1})\Delta^{s-t}(g_t)x^t\\
		=&f(x)-f_s\gamma^{s-t}(g_t^{-1})(\gamma^{s-t}(g_0)x^{s-t}+a_{n-1}x^{s-t-1}+\cdots+a_1x+\Delta^{s-t}(g_0))\\&
		-\cdots-f_sx^{s}-f_s\gamma^{s-t}(g_t^{-1})b_{s-t-1}x^{s-1}-\cdots-f_s\gamma^{s-t}(g_t^{-1})\Delta^{s-t}(g_t)x^t\\
		 =&f_0+f_1x+\cdots+f_sx^s-f_s\gamma^{s-t}(g_t^{-1})(\gamma^{s-t}(g_0)x^{s-t}+a_{s-t-1}x^{s-t-1}+\cdots\\&+a_1x+\Delta^{s-t}(g_0))-\cdots-f_sx^{s}-f_s\gamma^{s-t}(g_t^{-1})b_{s-t-1}x^{s-1}-\cdots-f_s\gamma^{s-t}(g_t^{-1})\\&\Delta^{s-t}(g_t)x^t\\
		 =&f_0+f_1x+\cdots+f_{s-1}x^{s-1}-f_s\gamma^{s-t}(g_t^{-1})(\gamma^{s-t}(g_0)x^{s-t}+a_{s-t-1}x^{s-t-1}\\&+\cdots+a_1x+\Delta^{s-t}(g_0))-\cdots-f_s\gamma^{s-t}(g_t^{-1})b_{s-t-1}x^{s-1}-\cdots-f_s\gamma^{s-t}(g_t^{-1})\\&\Delta^{s-t}(g_t)x^t
	\end{align*} where $a_1,a_2,\dots,a_{s-t-1},b_1,b_2,\dots,b_{s-t-1}\in \mathscr{R}_{q,s}$. Now, we can conclude that $\deg~h(x)<\deg~f(x)$. Hence, by induction on $\deg~h(x),$ there exist $b(x), r(x)\in \mathscr{R}_{q,s}[x;\gamma, \Delta]$ such that
	$$h(x)=b(x)g(x)+r(x),$$ where  $r(x)=0$ or $\deg~r(x)<\deg~g(x)$. Thus,
	\begin{align*}
		f(x)&=h(x)+f_s\gamma^{s-t}(g_t^{-1})x^{s-t}g(x) \\
		&=b(x)g(x)+r(x)+f_s\gamma^{s-t}(g_t^{-1})x^{s-t}g(x)\\
		&=(b(x)+f_s\gamma^{s-t}(g_t^{-1})x^{s-t})g(x)+r(x)\\
		&=q(x)g(x)+r(x),
	\end{align*}
	where $q(x)=b(x)+f_s\gamma^{s-t}(g_t^{-1})x^{s-t}\in  \mathscr{R}_{q,s}[x;\gamma, \Delta]$ and $r(x)=0$ or $\deg~r(x)<\deg~g(x)$. This gives the required result.
\end{proof}
Similarly,	one can define the left division algorithm. In above Theorem \ref{th algo}, if $r(x)=0,$ then $g(x)$ is called a right divisor of $f(x)$ or $f(x)$ is a left multiple of $g(x)$ in $\mathscr{R}_{q,s}[x;\gamma, \Delta]$. Throughout this paper, we consider the right division.

\section{ $(\theta,\Im)$-cyclic codes over $\mathbb{F}_q$}
This section presents the algebraic properties of $(\theta,\Im)$-cyclic codes in $R=\mathbb{F}_q[x;\theta,\Im]$ and provide a necessary and sufficient condition for these codes to contain their Euclidean duals. In \cite{MH}, Boulagouaz and Leroy introduced the notion of $(f,\gamma,\Delta)$-cyclic codes. Moreover, a $(\theta,\Im)$-cyclic code $\mathcal{C}$ is the subset of $\mathbb{F}_q^n$ consisting of the
coordinates of the elements of $Rg(x)/\langle x^n-1\rangle$ in the basis $\{1,x,\dots,x^{n-1}\}$ for some right
monic factors $g(x)$ of $x^n-1$.
\begin{theorem}\label{thm1}
	Let $g(x)=g_0+g_1x+\cdots+
	g_rx^r\in R$ be a monic polynomial.
	\begin{enumerate}
		\item A $(\theta,\Im)$-cyclic code of length $n$ corresponding to $Rg(x)/\langle x^n-1\rangle$ is a free left $\mathbb{F}_q$-module of dimension $n-r$ where $r=\deg~g(x)$.
		\item If $v=(v_0,v_1,\dots,v_{n-1})\in \mathcal{C},$ then $T_{\theta,\Im}(v)\in \mathcal{C}$.
		\item The rows of the matrix which generates the code $\mathcal{C}$ are given by
		$$T_{\theta,\Im}^k(g_0,g_1,\dots,g_r,0,0,\dots,0),~~~~\text{for}~0\leq k\leq n-r-1.$$
	\end{enumerate}
\end{theorem}
\begin{proof}
	\begin{enumerate}
		\item We have $x^n-1=h(x)g(x)$ for some monic polynomials $h(x)\in R$. Hence, as left $R$-module, we have $Rg(x)/\langle x^n-1\rangle\cong R/\langle h(x)\rangle$. Since $h$ is monic, $R/\langle h(x)\rangle$ is a free left $\mathbb{F}_q$-module of rank $deg~h(x)=n-r$.
		\item $v=(v_0,v_1,\dots,v_{n-1})\in \mathcal{C}$ if and only if $v(x):=\sum_{i=0}^{n-1} v_{i}x^{i}+\langle x^n-1\rangle\in Rg(x)/\langle x^n-1\rangle$. Since $xv(x)\in Rg(x)/\langle x^n-1\rangle$ and left multiplication by $x$ on $R/\langle x^n-1\rangle$ corresponds to the action of $T_{\theta,\Im}$ on $\mathbb{F}_q^n$, we have $T_{\theta,\Im}(v)\in \mathcal{C}$.
		\item We have $T_{\theta,\Im}^k(v_0,v_1,\dots,v_{n-1})\in \mathcal{C}$ for any $k\geq 0$. On the other hand, it is clear that $g, xg, x^2g,\dots, x^{n-r-1}g$ are left linearly independent over $\mathbb{F}_q$, all are taken modulo $x^n-1$ and hence form a basis of $Rg(x)/\langle x^n-1\rangle$. In codewords representation, this implies that the vectors $T_{\theta,\Im}^k(g_0,g_1,\dots,g_r,0,\dots,0)$ form a left $\mathbb{F}_q$-basis for $\mathcal{C}$, $0\leq k\leq n-r-1.$
	\end{enumerate}
\end{proof}
\begin{theorem}\label{l31}
	Let $\mathcal{C}$ be a left $R$-submodule of $R/\langle x^n-1\rangle$. Then $\mathcal{C}$ is a $(\theta,\Im)$-cyclic submodule generated by a monic polynomial of the smallest degree in $\mathcal{C}$.
\end{theorem}
\begin{proof}
	Let $g(x)\in \mathcal{C}$ be a monic smallest degree polynomial among nonzero polynomials in $\mathcal{C}$ and  $c(x)\in \mathcal{C}$. Then by Theorem \ref{th algo}, there exist unique polynomials $q(x)$ and $r(x)$ in $R$ such that
	$c(x) = q(x)g(x) + r(x)$ where $r(x)=0~ \text{or}~ \deg r(x)<\deg g(x).$ As $\mathcal{C}$ is a left $R$-submodule, we have $r(x)=c(x)-q(x)g(x)\in \mathcal{C}$. This is a contradiction to the assumption that $g(x)$ is of the smallest degree in $\mathcal{C}$ unless $r(x)=0$. This implies
	$c(x) = q(x)g(x)$ and hence $\mathcal{C}$ is a $(\theta,\Im)$-cyclic submodule generated by $g(x)$.
\end{proof}
\begin{theorem}
	Let $\mathcal{C}=\langle g(x) \rangle$ be a left $R$-submodule of $R/\langle x^n-1\rangle$, where $g(x)$ is a monic polynomial of smallest degree in $\mathcal{C}$. Then $g(x)$ is a right divisor of $x^n-1$.
\end{theorem}
\begin{proof}
	Consider a monic smallest degree polynomial $g(x)$ in $\mathcal{C}$. From Theorem \ref{th algo}, there exist polynomials $q(x)$ and $r(x)$ in $R$ such that
	$x^n-1 = q(x)g(x) + r(x)$, where $\deg r(x)<\deg g(x).$ Since $g(x)$ and $x^n-1=0$ are in $\mathcal{C}$, this implies $r(x)=(x^n-1)-q(x)g(x)\in \mathcal{C}$. But, $g(x)$ is smallest in $\mathcal{C}$. Therefore, $r(x)=0$ and hence $g(x)$ is a right divisor of $x^n-1$.
\end{proof}
Let $\mathcal{C}$ be a $(\theta,\Im)$-cyclic code of length $n$ over $\mathbb{F}_q$ generated by the right divisor $g(x)$ of $x^n-1$, where $g(x)=g_0+g_1x+\cdots+g_rx^r\in R$ and $g_r=1$. Then from the above discussion, we can conclude that  $\mathcal{C}$ is a free left $\mathbb{F}_q$-module of dimension $k=n-\deg g(x)$. Now, by using \cite[Theorem~3.2]{Leroy}, the generator matrix of $\mathcal{C}$ is given by
\begin{equation}\label{eq3}
	G=
	\begin{pmatrix}
		g\\
		T_{\theta,\Im}(g)\\
		\vdots\\
		T_{\theta,\Im}^{k-1}(g)
	\end{pmatrix}
\end{equation}
where $g=(g_0,g_1,g_2,\dots,g_{r})$ is the codeword corresponding to $g(x)$. Moreover, it is well known that $\dim(\mathcal{C})+\dim(\mathcal{C}^\perp)=n$. Therefore, $\dim(\mathcal{C}^\perp)=n-k=r$. Further, for our convenience, we define a one-to-one correspondence between the algebraic structures and combinatorial structures of $(\theta,\Im)$-cyclic codes as follows:
$$\tau:~~~~\mathbb{F}_q^n~~~~~\longrightarrow~~~~ \mathbb{F}_q[x;\theta,\Im]/\langle x^n-1\rangle$$
$$(c_0,c_1,c_2,\dots,c_{n-1})\longmapsto c_0+c_1x+c_2x^2+\dots+c_{n-1}x^{n-1}.$$
\begin{theorem}
	Let $\mathcal{C}=\langle g(x) \rangle$ be a $(\theta,\Im)$-cyclic code of length $n$ over $\mathbb{F}_q$, for some right divisor $g(x)$ of $x^n-1$. Let $x^n-1=h(x)g(x)=g(x)h'(x)$ for some monic skew polynomials $g(x),h(x),h'(x)\in R$. Then $c(x)\in \mathbb{F}_q[x;\theta,\Im]/\langle x^n-1\rangle$ is contained in $\mathcal{C}$ if and only if $c(x)h'(x)=0$ in $\mathbb{F}_q[x;\theta,\Im]/\langle x^n-1\rangle$.
\end{theorem}
\begin{proof}
	Let $c(x)\in \mathbb{F}_q[x;\theta,\Im]/\langle x^n-1\rangle$ be contained in $\mathcal{C}$. Then $c(x)=a(x)g(x)$ for some $a(x)\in R$. Now,
	\begin{align*}
		c(x)&=a(x)g(x) ~\text{for some} ~a(x)\in R \\
		c(x)h'(x)&=a(x)g(x)h'(x)=a(x)h(x)g(x)\\
		&=a(x)(x^n-1)=0 ~\text{in}~ \mathbb{F}_q[x;\theta,\Im]/\langle x^n-1\rangle.
	\end{align*}
	Conversely, let $c(x)h'(x)=0$ for some $c(x)$ in $\mathbb{F}_q[x;\theta,\Im]/\langle x^n-1\rangle$. Then $c(x)h'(x)=q(x)(x^n-1)$ for some $q(x)\in \mathbb{F}_q[x;\theta,\Im]/\langle x^n-1\rangle$. Also, $$c(x)h'(x)=q(x)(x^n-1)=q(x)h(x)g(x)=q(x)g(x)h'(x).$$ This implies that $c(x)=q(x)g(x)\in \langle g(x) \rangle=\mathcal{C}$ as $h'(x)$ is a nonzero polynomial.
\end{proof}
Now, with the help of the above-defined correspondence, the following theorem provides the generator matrix of the dual code $\mathcal{C}^\perp$ of $(\theta,\Im)$-cyclic code $\mathcal{C}$ of length $n$ over $\mathbb{F}_q$.
\begin{theorem}\label{thm34}
	Let $\mathcal{C}=\langle g(x) \rangle$ be a $(\theta,\Im)$-cyclic code of length $n$ over $\mathbb{F}_q$ for some right divisor $g(x)$ of $x^n-1$ and $x^n-1=h(x)g(x)=g(x)h'(x)$ for some monic skew polynomials $g(x),h(x),h'(x)\in R$. Then $\deg g(x)$ linearly independent columns of the matrix
	$$H=
	\begin{pmatrix}
		h'\\
		T_{\theta,\Im}(h')\\
		\vdots\\
		T_{\theta,\Im}^{n-1}(h')
	\end{pmatrix}$$
	form a basis of $\mathcal{C}^\perp$.
\end{theorem}
\begin{proof}
	Consider a $(\theta,\Im)$-cyclic code $\mathcal{C}$ of length $n$ over $\mathbb{F}_q$. Let $\mathcal{C}=\langle g(x) \rangle$ where $g(x)$ is a right divisor of $x^n-1$, and its leading coefficient is a unit. Then there exists $h(x)=h_0+h_1x+\cdots+h_kx^k\in \mathbb{F}_q[x;\theta,\Im]/\langle x^n-1\rangle$ such that $x^n-1=h(x)g(x)=g(x)h'(x)$. Now, for $c(x)=c_0+c_1x+\cdots+c_{n-1}x^{n-1}\in \mathcal{C}$, we have
	$$\tau(c(T_{\theta,\Im})(h'))=c(x)h'(x)=a(x)g(x)h'(x)=a(x)h(x)g(x)=a(x)(x^n-1)=0$$ for some $a(x)$ in $\mathbb{F}_q[x;\theta,\Im]/\langle x^n-1\rangle$ and $c(x)h'(x)$ is taken modulo $x^n-1$. This implies $c(T_{\theta,\Im})(h')=0$. Thus, $0=c(T_{\theta,\Im})(h')=c_0+ c_1T_{\theta,\Im}(h')+c_2T_{\theta,\Im}^2(h')+\cdots+c_{n-1}T_{\theta,\Im}^{n-1}(h').$ This shows that $(c_0,c_1,c_2,\dots,c_{n-1}).H=0$ for any $c=(c_0,c_1,c_2,\dots,c_{n-1})\in \mathcal{C}$. Also, $\tau(T_{\theta,\Im}^k(h'))=x^kh'(x)$ for $k=0,1,\dots, n-deg~ h'(x)-1=deg ~g(x)-1$ and hence $\{h',T_{\theta,\Im}(h'),T_{\theta,\Im}^2(h'),\dots,\\T_{\theta,\Im}^{r-1}(h')\}$ are linearly independent.
\end{proof}	
We now derive a necessary and sufficient condition for $(\theta,\Im)$-cyclic codes to contain their duals codes over $\mathbb{F}_q$.

\begin{theorem}\label{th2}
	Let $\mathcal{C}=\langle g(x) \rangle$ be a $(\theta,\Im)$-cyclic code of length $n$ over $\mathbb{F}_q$, for some right divisor $g(x)$ of $x^n-1$ and $x^n-1=h(x)g(x)=g(x)h'(x)$ for some monic skew polynomials $g(x),h(x),h'(x)\in R$.
	Then $\mathcal{C}^{\perp}\subseteq \mathcal{C}$ if and only if $h'(x)h'(x)$ is divisible by $x^n-1$ from the right.
\end{theorem}
\begin{proof}
	Let $\mathcal{C}=\langle g(x) \rangle$ be a $(\theta,\Im)$-cyclic code over $\mathbb{F}_q$ such that $\mathcal{C}^{\perp}\subseteq \mathcal{C}$. Note that $h'(x)\in \mathcal{C}^\perp$ and $\mathcal{C}^\perp\subseteq \mathcal{C}=\langle g(x) \rangle$. Thus, $h'(x)=p(x)g(x)$ for some $p(x)\in R$. Now, multiplying both sides by $h'(x)$ from right, we get
	$$h'(x)h'(x)=p(x)g(x)h'(x)=p(x)(x^n-1).$$ Hence, $h'(x)h'(x)$ is divisible by $x^n-1$ from the right.

	Conversely,	let $h'(x)h'(x)$ be divisible by $x^n-1$ from the right. Then $h'(x)h'(x)=b(x)(x^n-1)$ for some $b(x)\in R$. Now, consider $a(x)\in\mathcal{C}^\perp=\langle h'(x)\rangle$, then $a(x)=c(x)h'(x)$ for some $c(x)\in R$. Multiplying both sides by $h'(x)$ from right and using $h'(x)h'(x)=b(x)(x^n-1)$, we get
	\begin{align*}
		a(x)h'(x)=c(x)h'(x)h'(x)&=c(x)b(x)(x^n-1)\\
		&=c(x)b(x)h(x)g(x)
		=c(x)b(x)g(x)h'(x),\\
		&\Big(a(x)-c(x)b(x)g(x)\Big)h'(x)=0.
	\end{align*}
	As $h'(x)$ is a nonzero polynomial, we have $a(x)-c(x)b(x)g(x)=0$, which gives $a(x)=c(x)b(x)g(x)$. Therefore, $a(x)\in \mathcal{C}=\langle g(x) \rangle$. Thus, $\mathcal{C}^\perp\subseteq \mathcal{C}$.
\end{proof}
Here, we present an example to show the construction of $(\theta,\Im)$-cyclic codes over $\mathbb{F}_q$ with the help of our derived results.
\begin{example}
	Let $q=49,n=14$. In $\mathbb{F}_{49}$, the Frobenius automorphism $\theta:\mathbb{F}_{49}\longrightarrow \mathbb{F}_{49}$ is defined by $\theta(a)=a^7$ whereas the $\theta$-derivation $\Im$ is defined by $\Im(a)=w^2(\theta(a)-a)$ for all $a\in \mathbb{F}_{49}$. Therefore, $R=\mathbb{F}_{49}[x;\theta,\Im]$ is a skew polynomial ring. In $\mathbb{F}_{49}[x;\theta,\Im],$ we have
	\begin{align*}
		x^{14}-1=&(w^9x^{12} + 3x^{11} + w^{41}x^{10} + w^{13}x^9 + w^{37}x^8 + w^{47}x^7 + w^{18}x^5 +6x^4 + w^{38}x^3 \\&+ w^{18}x^2 + w^{28}x + w^{12})(w^{39}x^2 + w^3x + w^{17})=h(x)g(x)\\
		=&(w^{39}x^2 + w^3x + w^{17})(w^9x^{12} + 3x^{11} + w^{41}x^{10} + w^{13}x^9 + w^{37}x^8 + w^{47}x^7 \\&+ w^{33}x^5 +4x^4 + w^{17}x^3 + w^{37}x^2 + w^{13}x + w^{23})=g(x)h'(x).
	\end{align*}
	Consider $g(x)=w^{39}x^2 + w^3x + w^{17}$, $h(x)=w^9x^{12} + 3x^{11} + w^{41}x^{10} + w^{13}x^9 + w^{37}x^8 + w^{47}x^7 + w^{18}x^5 +6x^4 + w^{38}x^3 + w^{18}x^2 + w^{28}x+ w^{12}$ and $h'(x)=w^9x^{12} + 3x^{11} + w^{41}x^{10} + w^{13}x^9 + w^{37}x^8 + w^{47}x^7 + w^{33}x^5+4x^4 + w^{17}x^3 + w^{37}x^2 + w^{13}x + w^{23}$.
	Then, by Theorem \ref{thm34} and Equation \ref{eq3}, $\mathcal{C}$ is a $(\theta,\Im)$-cyclic codes over $\mathbb{F}_{49}$  of length $14$ which is generated by $g(x)$. The generator and parity check matrices of  $\mathcal{C}$ are given by Equation \ref{eq3} and  Theorem \ref{thm34}
 respectively.
 Since, $h'(x)h'(x)$ is divisible by $x^{14}-1$ from the right and hence the code $\mathcal{C}$ is also a dual-containing code, i.e., $\mathcal{C}^{\perp}\subseteq \mathcal{C}$.
 \end{example}

\section{ $(\gamma,\Delta)$-cyclic codes over $\mathscr{R}_{q,s}$}
In this section, our main focus is to discuss the algebraic properties of $(\gamma,\Delta)$-cyclic codes over $\mathscr{R}_{q,s}$ via decomposition over $\mathbb{F}_q$. To do so, we consider a linear code $\mathcal{C}$ of length $n$ over $\mathscr{R}_{q,s}$. Towards this, we define
$$\mathcal{C}_i=\left\{t_i\in \mathbb{F}_q^n~|~ \sum_{i=0}^{s}\zeta_it_i\in \mathcal{C}, ~\text{for some}~ t_0,t_1,\dots,t_{i-1},t_{i+1},\dots,t_s\in \mathbb{F}_q^n \right\}$$ for $0\leq i\leq s$. Then $\mathcal{C}_i$ is a linear code of length $n$ over $\mathbb{F}_q$ and $\mathcal{C}$ can be decomposed as
$$\mathcal{C}=\displaystyle\bigoplus_{i=0}^{s}\zeta_i\mathcal{C}_i.$$
Further, we consider a map $\gamma:\mathscr{R}_{q,s}\rightarrow \mathscr{R}_{q,s}$ defined by
$$ \gamma(r)=\sum_{i=0}^{s}\zeta_i\theta(r_i)$$ where $r=\sum_{i=0}^{s}\zeta_ir_i$ and $\theta\in \text{Aut}(\mathbb{F}_q)$ defined by $\theta(r_i)=r_i^{p^t}$ for all $r_i\in \mathbb{F}_q$.
Then $\gamma$ is an automorphism on $\mathscr{R}_{q,s}$. Next, we define a map $\Delta: \mathscr{R}_{q,s}\rightarrow \mathscr{R}_{q,s}$ such that
$$\Delta(r)=(1+v_1+v_2+\cdots+v_{r})(\gamma(r)-r)$$ where $r=\sum_{i=0}^{r}r_iv_i$ and $r_i\in \mathbb{F}_q$.
\begin{theorem}
	The above defined map $\Delta$ is a $\gamma$-derivation of $\mathscr{R}_{q,s}$.
\end{theorem}
\begin{proof}
	Let  $r,t\in\mathscr{R}_{q,s}$, we have
	\begin{align*}
		\Delta(r+t)&=(1+v_1+v_2+\cdots+v_{s})(\gamma(r+t)-(r+t)) \\
		&=(1+v_1+v_2+\cdots+v_{s})(\gamma(r)-r)+(1+v_1+v_2+\cdots+v_{s})(\gamma(t)-t)\\
		&=\Delta(r)+\Delta(t)
	\end{align*}
	and
	\begin{align*}
		\Delta(rt)=&(1+v_1+v_2+\cdots+v_{s})(\gamma(rs)-rt) \\
		=&(1+v_1+v_2+\cdots+v_{s})(\gamma(r)\gamma(t))-(1+v_1+v_2+\cdots+v_{s})rt\\
		 =&(1+v_1+v_2+\cdots+v_{s})(\gamma(r)\gamma(t))-(1+v_1+v_2+\cdots+v_{s})rt\\&+(1+v_1+v_2+\cdots+v_{s})\gamma(r)t-(1+v_1+v_2+\cdots+v_{s})\gamma(r)t\\
		=&(1+v_1+v_2+\cdots+v_{s})\gamma(r)(\gamma(t)-t)-(1+v_1+v_2+\cdots+v_{s})(r-\gamma(r))t\\
		=&(1+v_1+v_2+\cdots+v_{s})\gamma(r)(\gamma(t)-t)+(1+v_1+v_2+\cdots+v_{s})(\gamma(r)-r)t\\
		=&\Delta(r)t+\gamma(r)\Delta(t).
	\end{align*}
	Hence, $\Delta$ is a $\gamma$-derivation of $\mathscr{R}_{q,s}$.
\end{proof}
Further, with the help of the defined decomposition of $\mathcal{C}$, we discuss the algebraic properties of $(\gamma,\Delta)$-cyclic codes over $\mathscr{R}_{q,s}$.
\begin{theorem}\label{th31}
	Let $\mathcal{C}=\displaystyle\bigoplus_{i=0}^{s}\zeta_i\mathcal{C}_i$ be a linear code of length $n$ over $\mathscr{R}_{q,s}$ where $\mathcal{C}_i$ is a linear code of length $n$ over $\mathbb{F}_q$ for $i=0,1,2,\dots,s$. Then $\mathcal{C}$ is a $(\gamma,\Delta)$-cyclic code  of length $n$ over $\mathscr{R}_{q,s}$ if and only if $\mathcal{C}_i$ is a  $(\theta,\Im)$-cyclic code of length $n$ over $\mathbb{F}_q$ for $i=0,1,2,\dots,s$.
\end{theorem}
\begin{proof}
	Let $\mathcal{C}=\displaystyle\bigoplus_{i=0}^{s}\zeta_i\mathcal{C}_i$ be a $(\gamma,\Delta)$-cyclic code of length $n$ over $\mathscr{R}_{q,s}$ and $a^i=(a_0^i,a_1^i,\dots,a_{n-1}^i)\\\in \mathcal{C}_i$, for $0\leq i \leq s$. Consider $r_j=\sum_{i=0}^{s}\zeta_ia_j^i$  for $0\leq j\leq n-1$. Then $r=(r_0,r_1,\dots,r_{n-1})\in \mathcal{C}$ and $T_{\gamma,\Delta}(r)\in \mathcal{C}$. Again, we have $\gamma(r_j)=\sum_{i=0}^{s}\zeta_i\theta(a_j^i)$ and $\Delta(r_j)= \Delta(\sum_{i=0}^{s}\zeta_ia_j^i)=\Delta(\zeta_0a_j^0)+\Delta(\zeta_1a_j^1)+\cdots+\Delta(\zeta_0a_j^s)$ for  $0\leq j\leq n-1$. Also,
	\begin{align*}
		\Delta(\zeta_0a_j^0)&= \Delta(\zeta_0)a_j^0+\gamma(\zeta_0)\Im(a_j^0)\\
		&=\Big((1+v_1+\cdots+v_s)(\gamma(\zeta_0)-\zeta_0)\Big)a_j^0+\zeta_0\Im(a_j^0)\\
		&=\zeta_0\Im(a_j^0).
	\end{align*} Similarly, $\Delta(\zeta_ia_j^i)=\zeta_i\Im(a_j^i)$ for $i=1,2,\dots,s$ and $0\leq j\leq n-1$. Hence, $T_{\gamma,\Delta}(r)=\sum_{i=0}^{s}\zeta_iT_{\theta,\Im}(a^i)$. This implies that $T_{\theta,\Im}(a^i)\in \mathcal{C}_i$ for $i=0,1,2,\dots,s$. Thus, $\mathcal{C}_i$ is a  $(\theta,\Im)$-cyclic code of length $n$ over $\mathbb{F}_q$ for $i=0,1,2,\dots,s$.\\
	Conversely, suppose $\mathcal{C}_i$ is a $(\theta,\Im)$-cyclic code of length $n$ over $\mathbb{F}_q$. Let $r=(r_0,r_1,\dots,\\r_{n-1})\in \mathcal{C} $ where $r_j=\sum_{i=0}^{s}\zeta_ia_j^i$  for $0\leq j\leq n-1$.  Consider, $a^i=(a_0^i,a_1^i,\dots,a_{n-1}^i)$, for $0\leq i \leq s$. Then $a^i\in \mathcal{C}_i$ and also $T_{\theta,\Im}(a^i)\in \mathcal{C}_i$. Similar to the first part of the proof, we have
	$$\gamma(r_j)=\sum_{i=0}^{s}\zeta_i\theta(a_j^i)$$ and $$\Delta(r_j)= \Delta\Big(\sum_{i=0}^{s}\zeta_ia_j^i\Big)=\sum_{i=0}^{s}\zeta_i\Im(a_j^i)$$ for $i=0,1,2,\dots,s$ and $0\leq j\leq n-1$.
	Then \begin{align*}
		T_{\gamma,\Delta}(r)=\gamma(r)M + \Delta(r)=&\Big(\gamma(r_{n-1})+ \Delta(r_0),\gamma(r_{o})+ \Delta(r_1),\gamma(r_{1})+ \Delta(r_2),\dots,\gamma(r_{n-2})+\\& \Delta(r_{n-1})\Big)\\
		=&\sum_{i=0}^{s}\zeta_iT_{\theta,\Im}(a^i)
		\in \displaystyle\bigoplus_{i=0}^{s}\zeta_i\mathcal{C}_i=\mathcal{C}.
	\end{align*}
	Therefore, $\mathcal{C}$ is a $(\gamma,\Delta)$-cyclic code of length $n$ over $\mathscr{R}_{q,s}$.
\end{proof}

\begin{theorem}
	Let  $\mathcal{C}=\displaystyle\bigoplus_{i=0}^{s}\zeta_i\mathcal{C}_i$ be a $(\gamma,\Delta)$-cyclic code of length $n$ over $\mathscr{R}_{q,s}$. Then $\mathcal{C}=\langle \zeta_0g_0(x),\zeta_1g_1(x),\dots,\zeta_sg_s(x)\rangle$ and $|\mathcal{C}|= q^{(s+1)n-\sum_{i=0}^{s} \deg( g_i(x))}$, where $g_i(x)$ is a generator polynomial of $\mathcal{C}_i$ for $i=0,1,2,\dots,s$.
\end{theorem}	
\begin{proof}
	Let  $\mathcal{C}=\displaystyle\bigoplus_{i=0}^{s}\zeta_i\mathcal{C}_i$ be a $({\gamma,\Delta})$-cyclic code of length $n$ over $\mathscr{R}_{q,s}$. Then, by Theorem \ref{th31}, $\mathcal{C}_i$ is a $({\theta,\Im})$-cyclic code over $\mathbb{F}_q,$ for $i=0,1,2,\dots,s$. This implies that $\mathcal{C}_i=\langle g_i(x)\rangle \subseteq \mathbb{F}_q[x;\theta,\Im]/\langle x^n-1\rangle$ for $i=0,1,2,\dots,s$. Thus,
	$$\mathcal{C}=\left\{r(x)|r(x)=\sum_{i=0}^{s}\zeta_ig_i(x), g_i(x)\in \mathcal{C}_i\right\}.$$ Hence, $\mathcal{C}\subseteq \langle \zeta_0g_0(x),\zeta_1g_1(x),\dots,\zeta_sg_s(x)\rangle$.\\
	On the other hand, we consider $\zeta_0f_0(x)g_0(x)+\zeta_1f_1(x)g_1(x)+\cdots+\zeta_sf_s(x)g_s(x)\in \langle \zeta_0g_0(x),\\\zeta_1g_1(x),\dots,\zeta_sg_s(x)\rangle\subseteq \mathscr{R}_{q,s}[x;\gamma,\Delta]/\langle x^n-1\rangle$ where $f_i(x)\in \mathscr{R}_{q,s}[x;\gamma,\Delta]/\langle x^n-1\rangle$ for $i=0,1,2,\dots,s$. Then there exists $s_i(x)\in \mathbb{F}_q[x;\theta,\Im]/\langle x^n-1\rangle$ such that $\zeta_if_i(x)=\zeta_is_i(x)$ for $i=0,1,2,\dots,s$. This implies that $\langle \zeta_0g_0(x),\zeta_1g_1(x),\dots,\zeta_sg_s(x)\rangle\subseteq \mathcal{C}$. Thus, $\mathcal{C}=\langle \zeta_0g_0(x),\zeta_1g_1(x),\\\dots,\zeta_sg_s(x)\rangle$.
	Moreover, $|\mathcal{C}|=|\mathcal{C}_0||\mathcal{C}_1|\cdots|\mathcal{C}_s|=q^{n-\deg (g_0(x))}q^{n-\deg(g_1(x))}\cdots q^{n-\deg(g_s(x))}\\=q^{(s+1)n-\sum_{i=0}^{s} \deg( g_i(x))}$.
\end{proof}	
\begin{theorem}\label{thm dual contain}
	Let $\mathcal{C}=\displaystyle\bigoplus_{i=0}^{s}\zeta_i\mathcal{C}_i$ be a $(\gamma,\Delta)$-cyclic code of length $n$ over $\mathscr{R}_{q,s}$ and $x^n-1=h_i(x)g_i(x)=g_i(x)h'_i(x)$ for some monic skew polynomials $g_i(x),h_i(x),h'_i(x)\in\mathbb{F}_q[x;\theta,\Im] $ for $i=0,1,2,\dots,s$. Then $\mathcal{C}^\perp\subseteq \mathcal{C}$ if and only if $h_i'(x)h_i'(x)$ is divisible by $x^n-1$ from the right.
\end{theorem}	
\begin{proof}
	Let $h_i'(x)h_i'(x)$ be divisible by $x^n-1$ from the right for $i=0,1,2,\dots,s$. Then, by Theorem \ref{th2}, we have $\mathcal{C}_i^\perp\subseteq \mathcal{C}_i$, $i=0,1,2,\dots,s$. This implies that $\displaystyle\bigoplus_{i=0}^{s}\zeta_i\mathcal{C}_i^\perp\subseteq \displaystyle\bigoplus_{i=0}^{s}\zeta_i\mathcal{C}_i$. Hence, $\mathcal{C}^\perp\subseteq \mathcal{C}$.\\
	Conversely, let $\mathcal{C}^\perp\subseteq \mathcal{C}$ , then $\displaystyle\bigoplus_{i=0}^{s}\zeta_i\mathcal{C}_i^\perp\subseteq \displaystyle\bigoplus_{i=0}^{s}\zeta_i\mathcal{C}_i$. Now, considering modulo $\zeta_i$, we get $\mathcal{C}_i^\perp\subseteq \mathcal{C}_i$ for $i=0,1,2,\dots,s$. Thus, $h_i'(x)h_i'(x)$ is divisible by $x^n-1$ on the right for $i=0,1,2,\dots,s$.
\end{proof}
The next corollary is a direct consequence of the Theorem \ref{thm dual contain}.
\begin{corollary}
	Let $\mathcal{C}=\langle g(x)\rangle$ be a $(\gamma,\Delta)$-cyclic code of length $n$ over $\mathscr{R}_{q,s}$ and $x^n-1=h_i(x)g_i(x)=g_i(x)h'_i(x)$ for some monic skew polynomials $g_i(x),h_i(x),h'_i(x)\in\mathbb{F}_q[x;\theta,\Im]$. Then $\mathcal{C}^\perp\subseteq \mathcal{C}$ if and only if $\mathcal{C}_i^\perp\subseteq \mathcal{C}_i$ for $i=0,1,2,\dots,s$.
\end{corollary}
\section{Constructions of  quantum codes and comparison with the existing codes}
The quantum error-correcting codes play a pivotal role in quantum information theory. For a long time, it has been difficult to provide a satisfactory solution to the problem of protecting information from quantum noises. However, after the introduction of the first quantum error-correcting codes by Shor et al. \cite{Shor}, a stream of great developments has emerged in information theory. Let $H_q(\mathbb{C})$ be a $q$-dimensional Hilbert vector space. Then the set of $n$-fold tensor product $H_q^n(\mathbb{C})=\underbrace{H_q(\mathbb{C})\otimes H_q(\mathbb{C})\otimes\cdots  \otimes H_q(\mathbb{C})}_{n\rm\ times}$ is a $q^n$-dimensional Hilbert space. Here, a $q^k$ dimensional subspace of $H_q^n(\mathbb{C})$ is called a quantum code with parameters $[[n,k,d]]_q$ where $d$ is the minimum distance, and $k$ is the dimension of the quantum code. Also, $\mathcal{C}$ is dual-containing if  $\mathcal{C}^\perp\subseteq \mathcal{C}$. Moreover, in $1997$, the quantum Singleton bound for binary codes was introduced by
Knill and Laflamme \cite{Knill}. In $1998$, Calderbank et al. \cite{Calderbank98} provided the quantum Singleton bound for all codes over finite fields as $k+2d\leq n+2.$
A quantum code is said to be a quantum MDS code if it attains the Singleton bound.

In this section, we first briefly review the mathematical representation of the quantum states, the operators acting on these states, and then we construct quantum codes from $(\gamma,\Delta)$-cyclic codes over $\mathscr{R}_{q,s}$.
\subsection{Quantum states and operators over qudits}
\label{sub41}
 For a quantum system with $\Gamma$ levels, the state of a unit system, a qudit, is a superposition of $\Gamma$ basis states of the system given by
 \begin{align*}
\ket{\psi}_\Gamma &= \sum_{i=0}^{\Gamma-1} a_i\ket{i}_\Gamma, ~~ \text{ where } a_i\in \mathbb{C} \text{ and } \sum_{i=0}^{\Gamma-1} |a_i|^2 = 1,
 \end{align*}
where the subscript $\Gamma$ refers to dimension of the unit quantum system.  Also, $\ket{\psi}_\Gamma=[a_0~a_1~\dots ~a_{\Gamma-1}]^{\mathrm{T}}$ and  $\ket{i}_\Gamma = \mathbf{e}_{{(i+1)}}^{{(\Gamma)}}$, where $\mathbf{e}_{{(i+1)}}^{{(\Gamma)}}$ is a vector in $\mathbb{C}^\Gamma$ with the $(i+1)^{\mathrm{st}}$ element being $1$ and rest of the elements being $0$.

From the second postulate of quantum mechanics, the operators acting on a quantum system belong to the unitary group $\mathrm{U}(\Gamma)$, which is a subset of $\mathbb{C}^{\Gamma \times \Gamma}$. As the cardinality of $\mathrm{U}(\Gamma)$ is infinite, we represent its elements in terms of a basis of $\mathbb{C}^{\Gamma \times \Gamma}$.

  For $\Gamma=2$, the Pauli basis $\mathcal{P}$ is the popularly chosen unitary basis.
  \begin{equation*}
  \mathcal{P}\!=\!\left\{\!\mathrm{I} \!=\! \begin{bmatrix}
  1 &0\\0&1
  \end{bmatrix}\!,\mathrm{X} \!=\! \begin{bmatrix}
  0&1\\1&0
  \end{bmatrix}\!,\mathrm{Y} \!=\! \begin{bmatrix}
  0 &-\mathrm{i}\\\mathrm{i}&0
  \end{bmatrix}\!,\mathrm{Z} \!=\! \begin{bmatrix}
  1 &0\\0&-1
  \end{bmatrix}\right\},
  \end{equation*}
where $\mathrm{i} = \sqrt{-1}$.
The generalized version of the Pauli group for arbitrary $\Gamma$, known as the Weyl-Heisenberg group, is defined by
\begin{align}
 \mathcal{G}_{\Gamma}^{(g)} = \left\{\omega_\Gamma^l\mathrm{X}_\Gamma(a)\mathrm{Z}_\Gamma(b)|a,b,l \in\mathbb{Z}_\Gamma\right\},\nonumber
\end{align}
 where $\omega_\Gamma \!=\! \mathrm{e}^{\!\frac{\mathrm{i}2\pi}{\Gamma}}$, $\mathrm{X}_\Gamma(a)\!\ket{c}_\Gamma \!:=\!\! \ket{(a\!+\!c)~\mathrm{mod}~\Gamma}_\Gamma$, and $\mathrm{Z}_\Gamma(b)\!\ket{c}_\Gamma := \omega_\Gamma^{bc}\!\ket{c}_\Gamma$ for every $c \in \mathbb{Z}_\Gamma$.
 The generalized Pauli basis $\mathcal{P}$ \cite{Mie}
  \begin{align}
 \mathcal{G}_\Gamma = \left\{\mathrm{X}_\Gamma(a)\mathrm{Z}_\Gamma(b)|a,b \in\mathbb{Z}_\Gamma\right\}.\label{eqn:Gen_Pauli_Basis}
\end{align}
 is obtained by neglecting the phase $\omega_\Gamma^l$ in $\mathcal{G}_\Gamma$. The basis operator of the form $\mathrm{X}_\Gamma(a)\mathrm{Z}_\Gamma(b)$ is uniquely represented by a vector of length $2$ defined over ring $\mathbb{Z}_\Gamma$, namely $[a|b]_\Gamma$ as
\begin{align*}
\mathrm{X}_\Gamma(a)\mathrm{Z}_\Gamma(b) \equiv  [a|b]_\Gamma.
\end{align*}

Next, we define a trace operation over the field elements as follows:
\begin{definition}[\cite{Ketkar}]
The field trace $\mathrm{Tr}_{p^m/p}(\cdot)$ is an $\mathbb{F}_p$-linear function $\mathrm{Tr}_{p^m/p}: \mathbb{F}_{p^m} \rightarrow \mathbb{F}_{p}$, given by $\mathrm{Tr}_{p^m/p}(\kappa)=\sum_{i=0}^{m-1}\kappa^{p^i}$, where $\kappa \in \mathbb{F}_{p^m}$.
\end{definition}
The function $\mathrm{Tr}_{p^m/p}(\cdot)$ is said to be $\mathbb{F}_p$-linear as $\mathrm{Tr}_{p^m/p}(a\kappa+b\chi) = a\hspace{0.07cm}\mathrm{Tr}_{p^m/p}(\kappa)+b\hspace{0.07cm}\mathrm{Tr}_{p^m/p}(\chi)$, for all $a,b \in \mathbb{F}_p$ and $\kappa, \chi \in \mathbb{F}_{p^m}$. We note that for an element $b \in \mathbb{F}_p$, $\mathrm{Tr}_{p^m/p}(b) = b$.

The group that generates the operator basis for $\mathbb{C}^{p^m \times p^m}$ defined in terms of the field based representation of basis states is \cite{Lidar}
 \begin{align}
 \mathcal{G}_{p^m}^{(g)} \!=\! \begin{cases} \left\{\omega^l\mathrm{X}^{(p^m)}(\kappa)\mathrm{Z}^{(p^m)}(\chi)\bigg|\kappa,\!\chi \in \mathbb{F}_{p^m} \text{ and }l \!\in\!\mathbb{Z}_p\right\}, & \text{when characteristic $p$ is odd},\\
 \left\{\mathrm{i}^{g}\omega^l\mathrm{X}^{(p^m)}(\kappa)\mathrm{Z}^{(p^m)}(\chi)\bigg|\kappa,\!\chi \in \mathbb{F}_{p^m} \text{ and }g,l \!\in\!\mathbb{Z}_p\right\}, & \text{when characteristic $p$ is even},\\
 \end{cases} \nonumber
\end{align}
where $\omega = \mathrm{e}^{\frac{\mathrm{i}2\pi}{p}}$, $\mathrm{i}=\sqrt{-1}$,
\begin{align}
\mathrm{X}^{(p^m)}(\kappa)\ket{\theta}_{p^m} &:= \ket{\kappa+\theta}_{p^m},~~~~~\forall \theta \in \mathbb{F}_{p^m},\label{eqn:X_Definition_Qudit}\\
\mathrm{Z}^{(p^m)}(\chi)\ket{\theta}_{p^m} &:= \omega^{\mathrm{Tr}_{p^m/p}(\chi\theta)}\ket{\theta}_{p^m},~~~\forall \theta \in \mathbb{F}_{p^m}. \label{eqn:Z_Definition_Qudit}
\end{align}
We note that the factor $\mathrm{i}^{g}$ is included in the basis $\mathcal{G}_{p^m}$ when the characteristic $p$ is even as $\mathrm{iI}$ belongs to $\mathcal{P}$ and $\mathcal{P} = \mathcal{G}_{p^m}$ for $p=2$ and $m=1$.

The operator basis for $\mathbb{C}^{p^m \times p^m}$ is
\begin{align}
 \mathcal{G}_{p^m} \!=\! \left\{\mathrm{X}^{(p^m)}(\kappa)\mathrm{Z}^{(p^m)}(\chi)\bigg|\kappa,\!\chi \in \mathbb{F}_{p^m}\right\}, \label{eqn:FieldErrorBasis}
\end{align}

 From equations \eqref{eqn:X_Definition_Qudit} and \eqref{eqn:Z_Definition_Qudit}, $\mathrm{X}^{(p^m)}(\kappa)$ and $\mathrm{Z}^{(p^m)}(\chi)$ are given by
 \begin{align}
 \mathrm{X}^{(p^m)}(\kappa) &= \underset{\theta \in \mathbb{F}_{p^m}}{\sum}\ket{\kappa + \theta}\bra{\theta}\label{eqn:X_Expression_Qudit},\\
  \mathrm{Z}^{(p^m)}(\chi) &= \underset{\theta \in \mathbb{F}_{p^m}}{\sum}\omega^{\mathrm{Tr}_{p^m/p}(\chi\theta)}\ket{\theta}\bra{\theta}.\label{eqn:Z_Expression_Qudit}
 \end{align}
 The basis operator of the form $\mathrm{X}^{(p^m)}(\kappa)\mathrm{Z}^{(p^m)}(\chi)$ is uniquely represented by a vector of length $2$ defined over field $\mathbb{F}_{p^m}$, namely $[\kappa|\chi]_{p^m}$
\begin{align*}
\mathrm{X}^{(p^m)}(\kappa)\mathrm{Z}^{(p^m)}(\chi) \equiv  [\kappa|\chi]_{p^m}.
\end{align*}
The above defined operators will be used in Section \ref{sec5.3} during the encoding and error correction procedures of our proposed quantum codes.  In order to construct quantum error-correcting codes, we first derive a necessary and sufficient condition for $(\gamma,\Delta)$-cyclic codes to be dual containing. Note that a quantum code $[[n,k,d]]_q$ is said to be better than $[[n',k',d']]_q$ if any one of the following or both hold:
\begin{enumerate}
	\item $d>d'$ when the code rate $\frac{k}{n}=\frac{k'}{n'}$ (Larger distance with same code rate).
	\item $\frac{k}{n}>\frac{k'}{n'}$ when the distance $d=d'$ (Larger code rate with same distance).
\end{enumerate}
Next, we define a Gray map and study $\mathbb{F}_q$-images of $(\gamma,\Delta)$-cyclic codes. Let $GL_{s+1}(\mathbb{F}_q)$ be the set of all $(s+1)\times (s+1)$ invertible matrices over $\mathbb{F}_q$. Now, $\varphi:\mathscr{R}_{q,s}\longrightarrow \mathbb{F}_q^{s+1}$ define by $$\varphi(r)=(r_0,r_1,\dots,r_s)G,$$ where $r=\sum_{i=0}^{s}\zeta_ir_i\in \mathscr{R}_{q,s}$, $G\in GL_{s+1}(\mathbb{F}_q)$ such that $GG^T=kI_{s+1}$, $G^T$ is the transpose matrix of $G$, $k\in \mathbb{F}_q^*$ and $I_{s+1}$ is the identity matrix of order $s+1$. It is easy to check that $\varphi$ is a bijection and can be extended over $\mathscr{R}_{q,s}^n$ componentwise. If we define Gray distance for a linear code $\mathcal{C}$ by $d_G(\mathcal{C})=d_H(\varphi(\mathcal{C}))$, then $\varphi$ is a linear distance preserving map from $(\mathscr{R}_{q,s}^n,d_G)$ to $(\mathbb{F}_q^{n(s+1)},d_H)$, where $d_H$ is the Hamming distance in $\mathbb{F}_q$.
\begin{proposition}\label{prop1}
	The Gray map $\varphi$ is an $\mathbb{F}_q$-linear and distance preserving map from $\mathscr{R}_{q,s}^n$ (Gray distance) to $\mathbb{F}_q^{(s+1)n}$ (Hamming distance).
\end{proposition}
\begin{proof}
	Let $a=(a_0,a_1,\dots,a_{n-1}),~b=(b_0,b_1,\dots,b_{n-1})\in \mathscr{R}_{q,s}^n$, where $a_j= \sum_{i=0}^{s} \zeta_ia_j^i $, $b_j= \sum_{i=0}^{s} \zeta_ib_j^i $  for $j=0,1,\dots,n-1$ and $a_j^i,~b_j^i\in \mathbb{F}_q$. Then
	\begin{align*}
		\varphi(a+b)=&\varphi(a_0+b_0,a_1+b_1,\dots,a_{n-1}+b_{n-1})\\
		=&\varphi (\zeta_0(a_0^0+b_0^0)+\zeta_1(a_0^1+b_0^1)+\cdots+\zeta_{s}(a_0^{s}+b_0^{s}),\dots, \zeta_0(a_{n-1}^0+b_{n-1}^0)\\&+\zeta_1(a_{n-1}^1+b_{n-1}^1)+\cdots+\zeta_{s}(a_{n-1}^{s}+b_{n-1}^{s}))\\
		=&[(a_0^0+b_0^0,a_0^1+b_0^1,\dots,a_{0}^{s}+b_0^{s})G,\dots, (a_{n-1}^0+b_{n-1}^0,a_{n-1}^1+b_{n-1}^1,\dots,\\&a_{n-1}^{s}+b_{n-1}^{s})G]\\
		=&[(a_0^0,a_0^1,\dots,a_{0}^{s})G,\dots, (a_{n-1}^0,a_{n-1}^1,\dots,a_{n-1}^{s})G]+[(b_0^0,b_0^1,\dots,b_{0}^{s+1})G,\dots,\\& (b_{n-1}^0,b_{n-1}^1,\dots,b_{n-1}^{s})G]\\
		=&\varphi(a)+\varphi(b).
	\end{align*}
	Now, for any $\lambda\in \mathbb{F}_q$, we have
	\begin{align*}
		\varphi(\lambda a)&=\varphi(\lambda a_0,\lambda a_1,\dots,\lambda a_{n-1})\\
		&=\varphi (\lambda \zeta_0a_0^0+\lambda \zeta_1a_0^1+\cdots+\lambda \zeta_sa_0^s,\dots, \lambda \zeta_0a_{n-1}^0+\lambda \zeta_1a_{n-1}^1+\cdots+\lambda \zeta_sa_{n-1}^s)\\
		&=[(\lambda a_0^0,\lambda a_0^1,\dots,\lambda a_{0}^s)G,\dots, (\lambda a_{n-1}^0,\lambda a_{n-1}^1,\dots,\lambda a_{n-1}^s)G]\\
		&=[\lambda(a_0^0,a_0^1,\dots,a_{0}^s)G,\dots, \lambda(a_{n-1}^0,a_{n-1}^1,\dots,a_{n-1}^s)G]\\
		&=\lambda[(a_0^0,a_0^1,\dots,a_{0}^s)G,\dots, (a_{n-1}^0,a_{n-1}^1,\dots,a_{n-1}^s)G]\\
		&=\lambda\varphi(a).
	\end{align*}
	Moreover, $d_G(a,b)=\omega_G(a-b)=\omega_H(\varphi(a-b))=\omega_H(\varphi(a)-\varphi(b))=d_H(\varphi(a),\varphi(b))$. Hence, $\varphi$ is a distance preserving map.
\end{proof}
\begin{theorem}
	If $\mathcal{C}$ is an $[n,k,d_G]$ linear code over $\mathscr{R}_{q,s}$, then $\varphi(\mathcal{C})$ is a $[(s+1)n,k,d_H]$ linear code over $\mathbb{F}_q$.
\end{theorem}
\begin{proof}
	Follows directly from Proposition \ref{prop1} and the definition of the Gray map.
\end{proof}
The Gray map $\varphi$ preserves the orthogonality as shown in the
next result.
\begin{lemma}\label{lemdual}
	Let $\mathcal{C}$ be a $(\gamma,\Delta)$-cyclic code of length $n$ over $\mathscr{R}_{q,s}$. Then $\varphi(\mathcal{C})^\perp=\varphi(\mathcal{C}^\perp)$. Further, $\mathcal{C}$ is self-dual if and only if $\varphi(\mathcal{C})$ is self-dual.
\end{lemma}
\begin{proof}
	Let $c=(c_0,c_1,\dots,c_{n-1})\in \mathcal{C}$ and $d=(d_0,d_1,\dots,d_{n-1})\in \mathcal{C}^{\perp}$ where $a_j= \sum_{i=0}^{s} \zeta_ic_j^i $, $b_j= \sum_{i=0}^{s} \zeta_id_j^i $  for $j=0,1,\dots,n-1$ and $a_j^i,~b_j^i\in \mathbb{F}_q$. Now, $c\cdot d=\sum_{j=0}^{n-1}c_jd_j=0$ gives $\sum_{j=0}^{n-1}(c^0_jd^0_j+c^1_jd^1_j+\dots+c^s_jd^s_j)=0$. Again,
	\begin{align*}
		&\varphi(c)=[(c_0^0,c_0^1,\dots,c_0^s)G,\dots, (c_{n-1}^0,c_{n-1}^1,\dots,c_{n-1}^s)G]=(\alpha_0G,\dots,\alpha_{n-1}G)\\
		&\text{and}\\
		&\varphi(d)=[(d_0^0,d_0^1,\dots,d_0^s)G,\dots, (d_{n-1}^0,d_{n-1}^1,\dots,d_{n-1}^s)G]=(\beta_0G,\dots,\beta_{n-1}G),
	\end{align*}{}
	where $\alpha_j=(c_j^0,c_j^1,\dots,c_j^s)$ and $\beta_j=(d_j^0,d_j^1,\dots,d_j^s)$ for $0\leq j\leq n-1$ and $GG^T=k I_{s+1}$. Also,
	\begin{align*}
		\varphi(c)\cdot \varphi(d)=\varphi(c) \varphi(d)^T&=\sum_{j=0}^{n-1}\alpha_jGG^T\beta_j^T\\
		&=k\sum_{j=0}^{n-1}\alpha_j\beta_j^T\\
		&=k\sum_{j=0}^{n-1}(c^0_jd^0_j+c^1_jd^1_j+\dots+c^s_jd^s_j)=0.
	\end{align*}
	Since $c\in \mathcal{C}$ and $d\in \mathcal{C}^{\perp}$ are arbitrary, $\varphi(\mathcal{C}^{\perp})\subseteq (\varphi(\mathcal{C}))^{\perp}$. On the other hand, as $\varphi$ is a bijective linear map, $\mid \varphi(\mathcal{C}^{\perp})\mid= \mid (\varphi(\mathcal{C}))^{\perp}\mid$. Therefore, $\varphi(\mathcal{C}^{\perp})=(\varphi(\mathcal{C}))^{\perp}$.
\end{proof}
\subsection{CSS Code Framework}\label{sec:CSSCode}

Calderbank, Shor, and Steane \cite{Calderbank1} \cite{Steane} proposed a framework to construct quantum error correction codes over qubits from two classical binary codes $C_1$ and $C_2$ that satisfy $C_1^{\perp} \subset C_2$. This class of codes are called the \textit{Calderbank-Shor-Steane (CSS) codes}. The condition $C_1^{\perp} \subset C_2$ is called the \textit{dual-containing condition} of CSS codes.
By considering the two codes $C_1$ and $C_2$ to be the same code, i.e., $C_1=C_2$, we can construct quantum codes from dual-containing classical codes as $C_1^{\perp} \subset C_2 = C_1$.

The CSS codes form a class of stabilizer codes. Let $H_1$ and $H_2$ be the parity check matrices of the classical codes $C_1[n,k_1,d_1]$ and $C_2[n,k_2,d_2]$, respectively. As $C_1^{\perp} \subset C_2$, the elements of $C_1^{\perp}$ are codewords of $C_2$; hence, $H_2H_1^{\mathrm{T}} = 0$.

The CSS code is defined in the following two equivalent ways:
\begin{itemize}
\item[1)] The coset-based definition:
As $C_1^{\perp}\subset C_2$, cosets of $C_1^{\perp}$ are formed in $C_2$. The basis codewords of the CSS code $\mathcal{Q}_{\mathrm{CSS}}$ are the normalized superposition of all the elements in a particular coset of $C_1^{\perp}$ in $C_2$.
 As $C_1^{\perp}$ has $2^{(n-k_1)}$ elements and $C_2$ has $2^{k_2}$ elements, we obtain $C_2$ to contain $(2^{k_2})/(2^{(n-k_1)}) = 2^{(k_1+k_2-n)}$ cosets of $C_1^{\perp}$. As each coset corresponds to a basis codeword, $\mathcal{Q}_{\mathrm{CSS}}$ has a dimension of $2^{(k_1+k_2-n)}$.\vspace{0.1cm}

 Let $\omega_0$, $\omega_1$, $\dots$, $\omega_{g-1}$ be the cosets of $C_1^{\perp}$ in $C_2$, where $g=2^{(k_1+k_2-n)}$. Let $w_0$, $\dots$, $w_{g-1}$ be the coset representatives of the cosets $\omega_0$, $\dots$, $\omega_{g-1}$. The basis codeword of the CSS code corresponding to the coset $\omega_i$ ($i\in\{0,1,\dots,g-1\}$) is
 \begin{align}
\ket{\psi_i}=\frac{1}{2^{((n-k_1)/2)}}\sum_{l\in \mathcal{C}_1^\perp}\ket{l+w_i}. \label{eqn:CSS_Codeword}
 \end{align}
The basis states in the superposition help to detect/correct the bit flip errors while the superposition helps to detect/correct the phase flip errors.
 \item[2)] The check matrix based definition: The check matrix of the CSS code \cite{Nielsen} is
 \begin{align}
 \mathcal{H}_{\mathrm{CSS}} = \left[\begin{array}{c|c}
H_1 & \mathbf{0}\\
\mathbf{0} & H_2
\end{array}\right]. \label{eqn:H_CSS_b}
\end{align}
\end{itemize}
 The quantum codes obtained from both these definitions are the same for qubits.

Let $\rho_1 = (n-k_1)$ and $\rho_2=(n-k_2)$. From equation \eqref{eqn:H_CSS_b}, the first $\rho_1$ stabilizer generators that correspond to $[H_1 | 0]$ operate only the bit flip operator on a few qubits. They do not operate phase flip operators. As the bit flip and phase flip operators do not commute with each other, these stabilizers are used to detect and correct the phase flip errors. Similarly, the stabilizers that correspond to $[0|H_2]$ detect and correct the bit flip errors.

 As the stabilizer code \cite{Gottesman} correct bit flip errors and phase flip errors based on the stabilizers in $[0|H_2]$ and $[H_1 | 0]$, their bit flip and phase flip error correction capabilities are based on the error correction capabilities of $H_2$ and $H_1$, respectively. The minimum distance of the code is obtained to be $d' \geq \mathrm{min}(d_1,d_2)$ \cite{Nielsen}.

Suppose that the parity check matrices $H_1$ and $H_2$ are full rank matrices. The check matrix in equation \eqref{eqn:H_CSS_b} is a $((\rho_1+\rho_2)\times 2n)$ matrix. As $H_1$ and $H_2$ are full rank matrices, the CSS code has $(\rho_1 + \rho_2)$ minimal stabilizer generators. Thus, the size of the CSS code is $2^{(n-(\rho_1+\rho_2))} = 2^{(k_1+k_2-n)}$. Hence, the CSS code is an $[[n,k_1+k_2-n, d'\geq \mathrm{min}(d_1,d_2)]]$ stabilizer code.

Next, we discuss the CSS code over qudits that is obtained from the classical codes $D_1$ and $D_2$ by using two different  approaches for obtaining the basis codewords.
\begin{enumerate}
    \item \textbf{Coset-based construction of the CSS code:}
As $D_1^{\perp}$ is a subset of $D_2$, there exist cosets of $D_1^{\perp}$ in $D_2$. The size of $D_1^{\perp}$ and $D_2$ are $p^{m(n-k_1)}$ and $p^{mk_2}$, respectively; hence, the number of cosets of $D_1^{\perp}$ in $D_2$ is $s'=(p^{mk_2}/p^{m(n-k_1)}) = p^{m(k_2-n+k_1)} = p^{m(k_1+k_2-n)}$. Thus, the dimension of the quantum code obtained is $p^{m(k_1+k_2-n)}$, similar to the CSS code over qubits whose dimension is $2^{(k_1+k_2-n)}$.

Let $\tau_0$, $\tau_1$, $\dots$, $\tau_{(s'-1)}$ be the $s'$ cosets of $D_1^{\perp}$ in $D_2$. Let ${t_0}$, ${t_1}$, $\dots$, ${t_{(s'-1)}}$ be the coset representatives of $\tau_0$, $\tau_1$, $\dots$, $\tau_{(s'-1)}$, respectively. The basis codeword $\ket{\psi_i^{(p^m)}}$ ($i\in\{0,1,\dots,s'-1\}$) of the CSS code over qudits obtained from the coset $\tau_i$ is
 \begin{align}
\ket{\psi_i^{(p^m)}}=\frac{1}{p^{m((n-k_1)/2)}}\sum_{l\in \mathcal{D}_1^\perp}\ket{l+t_i}. \label{eqn:CSS_Codeword1}
 \end{align}
\item \textbf{Parity check matrix of the CSS code (\cite{Nadkarni2021}) :}
 The check matrix for the CSS code obtained from $D_1$ and $D_2$ that satisfy $D_1^{\perp}\subset
 D_2$, whose basis codewords are provided in Equation \ref{eqn:CSS_Codeword1}, is given by,\vspace{-0.15cm}
 \begin{align}
\mathcal{H}_{\mathrm{CSS}}^{(p^{m})} =  \left[\begin{array}{c|c}
\begin{matrix}
H_{d_1}\\\alpha H_{d_1} \\ \vdots \\ \alpha^{m-1}H_{d_1}
\end{matrix} & \mathbf{0}\\
\mathbf{0} & \begin{matrix}
H_{d_2}\\\alpha H_{d_2} \\ \vdots \\ \alpha^{m-1}H_{d_2}
\end{matrix}
\end{array}\right],
 \end{align}
 where $\alpha$ is the primitive element of $\mathbb{F}_{p^{m}}$.
\end{enumerate}
Now, keeping the above discussion in mind, we derive a necessary and sufficient condition for dual-containment. Currently, CSS construction (Lemma \ref{lemma css}) is one of the widely used techniques to obtain quantum codes from classical linear codes, in which dual containing linear codes play an instrumental role.
\begin{lemma}[\cite{Grassl04}, Theorem 3]\label{lemma css}
	Let $\mathcal{C}$ be an $[n,k,d]$ linear code over $\mathbb{F}_q$ such that $\mathcal{C}^{\perp}\subseteq \mathcal{C}$. Then there exists a quantum code with parameters $[[n,2k-n,d]]_q$.
\end{lemma}
\begin{theorem}\label{thm quatum}
	Let  $\mathcal{C}=\displaystyle\bigoplus_{i=0}^{s}\zeta_i\mathcal{C}_i$ be a $(\gamma,\Delta)$-cyclic code of length $n$ over $\mathscr{R}_{q,s}$.  Also, let $\mathcal{C}_i=\langle g_i(x) \rangle$ be a $(\theta,\Im)$-cyclic code over $\mathbb{F}_q$ where $x^n-1=h_i(x)g_i(x)=g_i(x)h'_i(x)$ for some monic skew polynomials $g_i(x),h_i(x),h'_i(x)\in\mathbb{F}_q[x;\theta,\Im],$ for $i=0,1,\dots,s$. Further, let $h_i'(x)h_i'(x)$ be divisible by $x^n-1$ from the right for  $i=0,1,\dots,s$. Then there exists a quantum code with parameters $[[(s+1)n,2k-(s+1)n,d_H]]_q$.
\end{theorem}
\begin{proof}
	Let $h_i'(x)h_i'(x)$ be divisible by $x^n-1$ from right for $i=0,1,\dots,s$. Then from Theorem \ref{thm dual contain}, we have $\mathcal{C}^{\perp}\subseteq \mathcal{C}$. Also, by Lemma \ref{lemdual},
	we have $\varphi(\mathcal{C}^{\perp})=\varphi(\mathcal{C})^{\perp}$, and hence $\varphi(\mathcal{C})^{\perp}\subseteq \varphi(\mathcal{C})$. Thus, $\varphi(\mathcal{C})$ is a dual containing linear code with parameters $[(s+1)n,k,d_H]$ over $\mathbb{F}_q$. Further, by Lemma \ref{lemma css}, there exists a quantum code with parameters $[[(s+1)n,2k-(s+1)n,d_H]]_q$.
\end{proof}
Next, with the help of our established results, we construct many new quantum codes possessing better parameters than the existing codes, which are appeared in \cite{Edel,Verma21}. In the following examples, $\mathbb{F}_q^*=\langle w \rangle $  denotes the cyclic group of non-zero elements of $\mathbb{F}_q$ generated by $w\in \mathbb{F}_q$. All examples' computations are carried out using the Magma
computation system \cite{Bosma}.
\begin{example}
	Let $q=8$, $s=3$ and $\mathscr{R}_{8,3}=\mathbb{F}_8[v_1,v_2,v_3]/\langle v_1^2-v_1, v_2^2-v_2,v_3^2-v_3,v_1v_2=v_2v_1 =v_2v_3 = v_3v_2 = v_3v_1 =v_1v_3 = 0\rangle$, where $\mathbb{F}_8=\mathbb{F}_2(w)$ and $w^3+w+1=0$. Let $n=30$,  $\theta:\mathbb{F}_8\longrightarrow \mathbb{F}_8$ be the Frobenius automorphism defined by $\theta(a)=a^2$, and the $\theta$-derivation $\Im:\mathbb{F}_8\longrightarrow \mathbb{F}_8$ is defined by $\Im(a)=w(\theta(a)-a)$ for all $a\in \mathbb{F}_8$. Therefore, $\mathbb{F}_8[x;\theta,\Im]$ is a skew polynomial ring. In $\mathbb{F}_8[x;\theta,\Im],$ we have
	\allowdisplaybreaks
	\begin{align*}
		x^{30}-1=&(w^6x^{29} + w^4x^{28} + w^6x^{27} + w^4x^{26} + w^3x^{25} + x^{24} + w^6x^{23} +
		w^4x^{22} + w^6x^{21}\\& + w^4x^{20}+ w^6x^{19} + w^4x^{18} + w^6x^{17} + w^4x^{16} +
		w^6x^{15} + w^4x^{14} + w^3x^{13} \\&+ x^{12} + w^6x^{11} + w^4x^{10} + w^3x^9 + x^8 +
		w^6x^7 + w^4x^6 + w^3x^5 + x^4 + w^6x^3\\& + w^4x^2+ w^3x + 1)(w^2x + 1)=h_0(x)g_0(x)\\
		=&(w^2x + 1)(w^6x^{29} + w^4x^{28} + w^6x^{27} + w^4x^{26} + w^6x^{25} + w^4x^{24} + w^6x^{23} \\&+
		w^4x^{22} + w^6x^{21} + w^4x^{20} + w^6x^{19} + w^4x^{18} + w^6x^{17} + w^4x^{16} +
		w^6x^{15} \\&+ w^4x^{14} + w^6x^{13} + w^4x^{12} + w^6x^{11} + w^4x^{10} + w^6x^9 +
		w^4x^8 + w^6x^7 \\&+ w^4x^6 + w^6x^5 + w^4x^4 + w^6x^3 + w^4x^2 + w^6x +
		w^4)=g_0(x)h_0'(x)\\
		x^{30}-1=&(w^5x^{28} + w^3x^{27} + w^2x^{26} + w^3x^{25} + w^3x^{24} + w^4x^{23} + w^6x^{22} +
		w^5x^{21} \\&+ w^6x^{20} + wx^{19} + w^3x^{18} + x^{17} + w^6x^{16} + w^3x^{15} + w^6x^{14}
		+ w^6x^{13} + w^4x^{12} \\&+ w^2x^{11} + w^6x^9 + w^5x^8 + w^6x^6 + w^4x^5 +
		w^5x^4 + w^4x^2 + w^2x + 1)(wx^2 \\&+ w^4x + w^6)=h_1(x)g_1(x)\\
		=&(wx^2 + w^4x + w^6)(w^5x^{28} + w^3x^{27} + w^2x^{26} + w^3x^{25} + x^{24} + w^3x^{22} + w^6x^{21} \\&+
		w^4x^{20} + x^{19} + w^5x^{18} + x^{16} + w^6x^{15} + w^5x^{13} + w^3x^{12} + w^2x^{11} +
		w^3x^{10} \\&+ x^9 + w^3x^7 + w^6x^6 + w^4x^5 + x^4 + w^5x^3 + x + w^6)=g_1(x)h_1'(x)\\
		x^{30}-1=&(w^6x^{28} + w^6x^{27} + wx^{26} + wx^{24} + x^{23} + w^5x^{22} + x^{20} + w^6x^{19} +
		w^5x^{17} \\&+ w^3x^{16} + w^2x^{15} + w^3x^{14} + x^{13} + w^4x^{12} + w^6x^9 + wx^8 +
		w^6x^7 + wx^6\\& + wx^5 + w^2x^4 + wx^3 + w^4x^2 + w^2x)(w^4x^2 + w^3x + w)=h_2(x)g_2(x)\\
		=&(w^4x^2 + w^3x + w)(w^6x^{28} + w^6x^{27} + wx^{26} + w^2x^{24} + wx^{23} + wx^{22} + w^3x^{21}\\& + w^2x^{20}
		+ w^2x^{19} + w^5x^{18} + w^5x^{17} + x^{16} + w^2x^{15} + w^6x^{13} + w^6x^{12} +
		wx^{11} \\&+ w^2x^9 + wx^8 + wx^7 + w^3x^6 + w^2x^5 + w^2x^4 + w^5x^3 +w^5x^2 + x + w^2)\\&=g_2(x)h_2'(x)\\
		x^{30}-1=&(x^{28} + w^4x^{27} + w^3x^{26} + w^3x^{25} + wx^{24} + w^4x^{23} + x^{22} + w^4x^{20} +
		x^{18} + x^{17} \\&+ w^5x^{16} + w^5x^{15} + w^4x^{14} + w^5x^{13} + w^6x^{11} + w^5x^{10} +
		wx^9 + w^4x^8 + x^7\\& + wx^6 + wx^5 + w^3x^3 + w^6x^2 + w^6x + 1)(x^2 + w^2x + w^4)=h_3(x)g_3(x)\\
		=&(x^2 + w^2x + w^4)(x^{28} + w^4x^{27} + w^3x^{26} + w^3x^{25} + wx^{24} + w^2x^{23} + w^5x^{22}\\& + w^2x^{21}
		+ w^4x^{20} + wx^{19} + w^6x^{18} + x^{17} + w^5x^{16} + w^6x^{15} + x^{13} + w^4x^{12} \\&+
		w^3x^{11} + w^3x^{10} + wx^9 + w^2x^8 + w^5x^7 + w^2x^6 + w^4x^5 + wx^4 +
		w^6x^3\\& + x^2 + w^5x + w^6)=g_3(x)h_3'(x)
	\end{align*}
	Now, let $g_0=w^2x + 1,  g_1=wx^2 + w^4x + w^6$, $g_2=w^4x^2 + w^3x + w$ and $g_3=x^2 + w^2x + w^4$. Then $\mathcal{C}_i=\langle g_i(x)\rangle$ is a $(\theta,\Im)$-cyclic code of length $30$ over $\mathbb{F}_8$ for $i=0,1,2,3$. Then by Theorem \ref{th31}, $\mathcal{C}=\displaystyle\bigoplus_{i=0}^{s}\zeta_i\mathcal{C}_i$ is a $(\gamma,\Delta)$-cyclic code of length $30$ over $\mathscr{R}_{8,3}$.
	Let 	
	\begin{equation*}\label{eq31}
		G=
		\begin{pmatrix}
			1&w&w^3&1\\
			w&1&1&w^3\\
			w^3&1&1&w\\
			1&w^3&w&1\\
		\end{pmatrix}\in GL_4(\mathbb{F}_8)
	\end{equation*}
	such that $GG^T=I_4$. Then $\varphi(\mathcal{C})$ is a $[120,114,4]$ linear code over $\mathbb{F}_8$. Again,
	\begin{align*}
		h_0'(x)h_0'(x)=&(w^2x^{28} + x^{27} + x^{26} + w^5x^{25} + w^4x^{24} + w^2x^{22} + x^{21} + x^{20} + w^5x^{19}\\&
		+ w^2x^{18} + x^{17}+ x^{16} + w^5x^{15} + w^2x^{14} + x^{13} + w^3x^{12} + w^4x^{11} +
		x^{10 }\\&+ w^5x^9 + w^4x^8 + w^2x^6 + x^5 + w^3x^4 + w^4x^3 + x^2 + w^5x + w^4)(x^{30}-1)\\
		h_1'(x)h_1'(x)=&(wx^{26} + w^4x^{25} + w^2x^{24} + wx^{23} + wx^{22} + w^3x^{21} + w^5x^{19} + w^4x^{16}\\&
		+ x^{15}+ w^2x^{14} + w^3x^{13} + wx^{12} + w^4x^{11} + w^6x^{10} + w^6x^9 + w^2x^8
		+ wx^7 \\&+ w^4x^6 + x^4 + x^2 + w^6x + w^6)(x^{30}-1)\\
		h_2'(x)h_2'(x)=&(w^4x^{26} + w^3x^{25 }+ w^3x^{24} + w^4x^{23} + x^{22} + w^3x^{21} + w^4x^{20} +
		w^3x^{19} \\&+ w^2x^{18}+ wx^{17} + wx^{16} + w^5x^{15} + x^{14} + w^5x^{13} + w^5x^{12} +
		w^2x^{11} \\&+ w^6x^9 + w^6x^8 + w^2x^7 + w^5x^6 + x^5 + w^4x^4 + w^6x^3 +
		w^5x^2 + w^5x)\\&(x^{30}-1)\\
		h_3'(x)h_3'(x)=&(x^{26} + w^2x^{25} + w^5x^{24} + x^{23} + w^2x^{22} + wx^{20} + x^{19} + w^3x^{18} +
		w^6x^{17} \\&+ wx^{16} + x^{15} + wx^{14} + w^6x^{12} + w^6x^{11} + w^5x^{10} + x^9 +
		w^3x^8 + x^7 \\&+ x^6 + w^6x^5 + w^6x^3 + w^4x^2 + w^6x + w^6)(x^{30}-1).
	\end{align*}
	From above, we see that $h_i'(x)h_i'(x)$ is divisible by $(x^{30}-1)$ on the right for $i=0,1,2,3$. Hence, by Theorem \ref{thm quatum}, there exists a quantum code with parameters $[[120,108,4]]_8$ which has the same length and distance but better code rate than the best-known code $[[120,104,4]]_8$ given by \cite{Edel}.
\end{example}

\begin{example}
	Let $q=25$, $s=3$ and $\mathscr{R}_{25,3}=\mathbb{F}_{25}[v_1,v_2,v_3]/\langle v_1^2-v_1, v_2^2-v_2,v_3^2-v_3,v_1v_2=v_2v_1 =v_2v_3 = v_3v_2 = v_3v_1 =v_1v_3 = 0\rangle$. Let $n=30$,  $\theta:\mathbb{F}_{25}\longrightarrow \mathbb{F}_{25}$ be the Frobenius automorphism defined by $\theta(a)=a^5$, and the $\theta$-derivation $\Im:\mathbb{F}_{25}\longrightarrow \mathbb{F}_{25}$ is defined by $\Im(a)=w(\theta(a)-a)$ for all $a\in \mathbb{F}_{25}$. Therefore, $\mathbb{F}_{25}[x;\theta,\Im]$ is a skew polynomial ring. In $\mathbb{F}_{25}[x;\theta,\Im],$ we have
	\allowdisplaybreaks
	\begin{align*}
		x^{20}-1=&(w^{19}x^{19} + x^{18} + w^{20}x^{17} + w^4x^{16} + w^{15}x^{15} + w^{20}x^{14} + x^{13} +
		w^{19}x^{12} + w^7x^{11} \\&+ w^2x^{10} + w^{10}x^9 + 3x^8 + w^3x^7 + w^{17}x^6 +w^{11}x^5 + 2x^4 + 3x^2 + 4x + 3)(wx \\&+ w^{17})=h_0(x)g_0(x)\\
		=&(wx + w^{17})(w^{19}x^{19} + x^{18} + w^{20}x^{17} + w^4x^{16} + w^{15}x^{15} + wx^{14} + 2x^{13} + w^2x^{12}\\&
		+ w^{10}x^{11} + w^{21}x^{10} + w^7x^9 + 4x^8 + w^8x^7 + w^{16}x^6 + w^3x^5 +
		w^{13}x^4 + 3x^3\\& + w^{14}x^2+ w^{22}x + w^9)=g_0(x)h_0'(x)\\
		x^{20}-1=&(w^{14}x^{18} + w^8x^{17} + w^{17}x^{16} + 3x^{15} + 2x^{14} + w^{21}x^{13} + w^8x^{12} +
		w^{10}x^{11} + wx^{10}\\& + 4x^9 + w^{19}x^7 + w^{19}x^6 + w^9x^5 + 2x^4 + 4x^3 + x +2)(w^{10}x^2+ 2x+w^{11})\\&
		=h_1(x)g_1(x)\\
		=&(w^{10}x^2+ 2x+w^{11})(w^{14}x^{18 }+ w^8x^{17} + w^{17}x^{16} + 3x^{15} + 2x^{14} + w^{15}x^{13} \\&+ w^3x^{12} + x^{11}
		+ w^{15}x^{10} + 3x^9 + w^4x^8 + 4x^7 + w^4x^6 + w^{10}x^5 + x^4 + 2x^3 \\&+w^{11}x^2 + w^{19}x+ w^{19}) =g_1(x)h_1'(x)\\
		x^{20}-1=&(w^{23}x^{19} + w^{19}x^{18} + w^3x^{17} + w^{14}x^{16} + x^{15} + w^4x^{14} + w^{19}x^{13} +
		w^{10}x^{12} \\&+ w^{13}x^{11} + 2x^{10} + w^2x^9 + w^{13}x^8 + w^7x^7 + w^{21}x^6 + 4x^5
		+ 2x^4 + 3x^2 + 4x \\&+ 3)(w^5x + 3)=h_2(x)g_2(x)\\
		=&(w^5x + 3)(w^{23}x^{19 }+ w^{19}x^{18} + w^3x^{17} + w^{14}x^{16} + x^{15} + w^5x^{14} + wx^{13} +\\&
		w^9x^{12} + w^{20}x^{11} + 2x^{10} + w^{11}x^9 + w^7x^8 + w^{15}x^7 + w^2x^6 + 4x^5
		+ w^{17}x^4 \\&+ w^{13}x^3 + w^{21}x^2 + w^8x + 3)=g_2(x)h_2'(x)\\
		x^{20}-1=&(w^{10}x^{18} + 4x^{17} + w^{11}x^{16} + 4x^{15} + w^{16}x^{14} + w^7x^{13} + 3x^{12} +
		w^8x^{11} \\&+ w^{21}x^{10} + w^{13}x^9 + 3x^8 + 2x^7 + w^{14}x^6 + w^2x^5 + 4x^4 +4x^3 + x^2 + 4x\\& + 2)(w^{14}x^2 + w^{19}x + w^{15})=h_3(x)g_3(x)\\
		=&(w^{14}x^2 + w^{19}x + w^{15})(w^{10}x^{18} + 4x^{17} + w^{11}x^{16} + 4x^{15} + w^{16}x^{14 }+ w^2x^{13}\\& + x^{12} + 3x^{11} +
		w^4x^{10} + w^4x^9 + w^7x^8 + 2x^7 + w^{22}x^6 + w^9x^5 + w^{10}x^4 \\&+ x^3 +w^{13}x + w^2
		)=g_3(x)h_3'(x)
	\end{align*}
	Now, let $g_0(x)=wx + w^{17},  g_1(x)=w^{10}x^2 + 2x + w^{11}$, $g_2(x)= w^5x + 3$ and $g_3(x)=w^{14}x^2 + w^{19}x + w^{15}$. Then $\mathcal{C}_i=\langle g_i(x)\rangle$ is a $(\theta,\Im)$-cyclic code of length $20$ over $\mathbb{F}_{25}$ for $i=0,1,2,3$. Then by Theorem \ref{th31}, $\mathcal{C}=\displaystyle\bigoplus_{i=0}^{s}\zeta_i\mathcal{C}_i$ is a $(\gamma,\Delta)$-cyclic code of length $20$ over $\mathscr{R}_{{25},3}$.
	Let 	
	\begin{equation*}\label{eq31}
		G=
		\begin{pmatrix}
			-1&1&1&1\\
			1&1&1&-1\\
			1&-1&1&1\\
			1&1&-1&1\\
		\end{pmatrix}\in GL_4(\mathbb{F}_{25})
	\end{equation*}
	such that $GG^T=4I_2$. Then $\varphi(\mathcal{C})$ is a $[80,74,4]$ linear code over $\mathbb{F}_{25}$. Again,
	\begin{align*}
		h_0'(x)h_0'(x)=&(3x^{18} + w^4x^{17} + w^{19}x^{16} + w^2x^{15} + w^{23}x^{13} + w^8x^{12} + w^9x^{11} +
		w^{11}x^{10}\\& + w^{21}x^8 + wx^7 + w^5x^6 + 3x^5 + wx^3 + 2x^2 + 3x + w^{15})(x^{20}-1)\\
		h_1'(x)h_1'(x)=&(w^4x^{16} + w^{13}x^{15} + w^{14}x^{14} + w^2x^{13} + w^{17}x^{12} + w^9x^{11} + w^{20}x^{10} +
		w^7x^9 \\&+ x^8 + w^{14}x^7 + w^{22}x^6 + w^{21}x^5 + w^8x^4 + 4x^3 + w^{22}x^2 +wx + w^{13})\\&(x^{20}-1)\\
		h_2'(x)h_2'(x)=&(3x^{18} + w^{20}x^{17} + w^{10}x^{16} + 2x^{15} + w^{19}x^{13} + 4x^{12} + w^{21}x^{11} +
		4x^{10}\\& + w^9x^8 + w^{15}x^7 + 2x^6 + 3x^5 + w^5x^3 + wx^2 + w^{13}x + 1)(x^{20}-1)\\
		h_3'(x)h_3'(x)=&(w^{20}x^{16} + w^{23}x^{15} + 2x^{14} + w^{14}x^{13} + w^{21}x^{12} + w^{13}x^{11} + w^{16}x^{10} \\&+w^{11}x^9 + w^5x^8 + w^{11}x^7 + w^{11}x^6 + w^{17}x^5 + w^8x^4 + w^8x^3 +
		w^{11}x^2 \\&+ w^{16}x + w^{20})(x^{20}-1).
	\end{align*}
	From above we see that $h_i'(x)h_i'(x)$ is divisible by $x^{20}-1$ on the right for $i=0,1,2,3$. Hence, by Theorem \ref{thm quatum}, there exists a quantum code with parameters $[[80,68,4]]_{25}$ which has the same length and distance, but better code rate than the best-known code $[[80,64,4]]_{25}$ given by \cite{Verma21}.
\end{example}

Let $\mathcal{C}$ be a $(\theta,\Im)$-cyclic code of length $n$ over $\mathbb{F}_q$ where $\mathcal{C}=\langle g(x)\rangle$ and $x^n-1=h(x)g(x)=g(x)h'(x)$ for some monic skew polynomials $g(x),h(x),h'(x)\in\mathbb{F}_q[x;\theta,\Im]$. Further, let $h'(x)h'(x)$ be divisible by $x^n-1$ from the right. Therefore, by Theorem \ref{thm dual contain}, we get the dual containing codes with $\mathbb{F}_q$-parameters $[n,k,d]_q$ (enlisted in the fourth column of Table \ref{tab2}).
\begin{landscape}
	\begin{table}
		\renewcommand{\arraystretch}{1.8}
		\begin{center}
			\caption{New quantum codes from $(\gamma,\Delta)$-cyclic codes over $\mathscr{R}_{q,s}$}	
			\begin{tabular}{|c|c|c|c|c|c|c|c|c|c|c|c|c|c|c|c|}
				%{p{1cm} p{.5cm} p{4cm} p{4cm} p{1.5cm} p{1.5cm}}
				
				%\multicolumn{6}{c}{Table 1: Linear codes as Gray images of $\lambda$-constacyclic codes} \\\\
				\hline
				$s$ & $(n,q)$ &  $\Im(a)$, $a\in \mathbb{F}_q$&$[g_0(x),g_1(x),\dots,g_s(x)]$ & $\varphi(\mathcal{C})$ & Obtained  & Existing  \\
				
				&   &  &  &   & Codes  & Codes \\
				
				\hline

				$2$ & $(48,9)$ &$w^2(\theta(a)-a)$& $(w^71w^3,w^5w^2,w^512)$   & $[144,139,3]_9$ & $[[144,134,3]]_9$ & $[[146,134,3]]_9$ \cite{Edel} \\
				
				$3$ & $(36,9)$ & $w^2(\theta(a)-a)$&$(w^21w^5,1w^3, w^7w^2,2ww^3)$   & $[144,138,4]_9$ & $[[144,132,4]]_9$ & $[[146,128,4]]_9$ \cite{Edel} \\
				
				$3$ & $(32,9)$ & $w^2(\theta(a)-a)$&$(w^2ww, w^6w^22, 2w^2ww^3,w^31)$   & $[128,120,4]_9$ & $[[128,112,4]]_9$ & $[[129,103,4]]_9$ \cite{Edel} \\
				
				$2$ & $(42,9)$ & $w^2(\theta(a)-a)$&$(ww^6, w^6ww^3,w^7w^2)$   & $[126,122,3]_{9}$ & $[[126,118,3]]_{9}$ & $[[130,118,3]]_{9}$ \cite{Verma21} \\
				
				$2$ & $(60,4)$ & $w(\theta(a)-a)$&$(www^2,1w^2w^2, 11w^2)$   & $[180,174,3]_{4}$ & $[[180,168,3]]_{4}$ & $[[185,167,3]]_{4}$ \cite{Verma21} \\
				
				$3$ & $(20,25)$ & $w(\theta(a)-a)$&$(ww^{17}, w^{10}2w^{11}, w^53,w^{14} w^{19}w^{15})$   & $[80,74,4]_{25}$ & $[[80,68,4]]_{25}$ & $[[80,64,4]]_{25}$ \cite{Verma21} \\
				
				$3$ & $(40,25)$ & $w(\theta(a)-a)$&$(w^{19}w^{10},w^{10}w^{14},w^{11} w^{17},w^{14}w^4)$   & $[120,116,3]_{25}$ & $[[120,112,3]]_{25}$ & $[[120,106,3]]_{25}$ \cite{Verma21} \\
				
				$3$ & $(30,8)$ & $w(\theta(a)-a)$&$(w^21, ww^4w^6, w^4w^3w,1w^2w^4)$   & $[120,114,4]_8$ & $[[120,108,4]]_8$ & $[[120,104,4]]_8$ \cite{Edel} \\
				$3$ & $(32,8)$ & $w(\theta(a)-a)$&$(w^6w^3w^4, w^2w^5w^5, ww^6,w^6w^5w^2)$   & $[128,121,4]_8$ & $[[128,114,4]]_8$ & $[[128,112,4]]_8$ \cite{Edel} \\
				\hline
			\end{tabular}\label{tab2}
		\end{center}
	\end{table}	
\end{landscape}
\noindent Also, by Lemma \ref{lemma css}, we construct quantum codes $[[n,k,d]]_q$ (in the fifth column), in which some codes satisfy the equality $n-k+2-2d=2$ (Near to MDS), and some are MDS (maximum-distance-separable). Let $\mathcal{C}=\displaystyle\bigoplus_{i=0}^{s}\zeta_i\mathcal{C}_i$ be a $(\gamma,\Delta)$-cyclic code of length $n$ over $\mathscr{R}_{q,s}$ where $\mathcal{C}_i=\langle g_i(x)\rangle$ is a $(\theta,\Im)$-cyclic code of length $n$ over $\mathbb{F}_q$ and $x^n-1=h_i(x)g_i(x)=g_i(x)h'_i(x)$ for some monic skew polynomials $g_i(x),h_i(x),h'_i(x)\in\mathbb{F}_q[x;\theta,\Im]$ for $i=0,1,2,\dots,s$. Further, let $h_i'(x)h_i'(x)$ is divisible by $x^n-1$ from the right for $i=0,1,2,\dots,s$. Therefore, by Theorem \ref{thm dual contain}, we get the dual containing codes with $\mathbb{F}_q$-parameters $[n,k,d]_q$ (enlisted in the fifth column of Table \ref{tab2}). Also, by Theorem \ref{thm quatum}, we construct quantum codes $[[n,k,d]]_q$ (in the sixth column), which beat the parameters of best-known codes (in the seventh column) given by the online database \cite{Edel,Verma21}. Also, the first and second columns represent $s$ and $(n,q)$, respectively. Moreover, in third column we present $\theta$-derivations $\Im(a)$ for $a\in\mathbb{F}_q$. Note that in fourth column we give generator polynomials $g_i$ for $\mathcal{C}_i$ ($i=0,1,2,\dots,s$) which is a right factor of $x^n-1$ in $\mathbb{F}_q[x;\theta,\Im]$. In order to make Table \ref{tab2} precise, we enlist the coefficients of polynomials in decreasing powers of $x$. For example, we write $w^702w$ to represent the polynomial $w^7x^3+2x+w$.
\subsection{Theory Behind the Encoding and Error Correction Procedure}\label{sec5.3}
\subsubsection{Encoding}
The dimension of the quantum code is $q^{2k-(s+1)n}$; hence, the remaining dimension \linebreak $q^{2(s+1)n-2k}$ corresponds to the redundancy. Let $\ket{\phi'}$ be the $(2k-(s+1)n)$ qudit message state.  We consider $(2(s+1)n-2k)$ qudits in state $\ket{0}$ each called the \textit{ancillary qudits} or \textit{ancilla qudits} that correspond to the redundancy that is added to the code. The encoding of the stabilizer quantum codes involves applying an operator $\mathcal{E}$ to the state $\ket{\phi'}\ket{0}^{\otimes (2(s+1)n-2k)}$. The encoding operator $\mathcal{E}$ is a product of operators from a group called the Clifford group. While working with basis operators, we need unitary operators that transform a basis operator to another basis operator,
called the Clifford operators \cite{JM}. The Clifford operators transform every Pauli basis operator into a Pauli basis operator. The
set of all Clifford operators forms the Clifford group that is
generated by the discrete Fourier transform $(\text{DFT}_q)$ operator,
phase shift operator, and the addition $(\text{ADD}_q)$ operator \cite{JM,Grassl03}.
\subsubsection{Error Correction}
Syndrome computation involves computing the syndrome based on the erroneous state $E\ket{\psi}$, where $E$ is an error that belongs to the Pauli basis $\mathcal{P}^{\otimes (s+1)n}$. We apply the syndrome computation operator that operates on $E\ket{\psi}$ along with $(2(s+1)n-2k)$  syndrome qudits in state $\ket{0}$ to transform it to $E\ket{\psi}\ket{s}$. Using the syndrome state $\ket{s}$ as the control and the codeword qudits as the target, the inverse error operation $E^{\dagger}$ is applied to obtain the codeword $\ket{\psi}$.

We next discuss the syndrome computation and error correction procedure when the error $E$ does not belong to the Pauli basis. The error $E$ belongs to $\mathbb{C}^{q^{(s+1)n}\times q^{(s+1)n}}$ as it is an $(s+1)n$ qudit operator; hence, $E$ can be expressed in terms of the Pauli basis $\mathcal{P}^{\otimes (s+1)n}$ as the Pauli basis is a basis for $\mathbb{C}^{q^{(s+1)n}\times q^{(s+1)n}}$.  Let
\begin{align}
&E = \sum_{B\in\mathcal{P}^{\otimes (s+1)n}} a_B B, \text{ where }a_B \in \mathbb{C},\nonumber\\&
\text{this implies}~~ E\ket{\psi} = \left(\sum_{B\in\mathcal{P}^{\otimes (s+1)n}} a_B B\right)\ket{\psi} = \sum_{B\in\mathcal{P}^{\otimes (s+1)n}} a_B B\ket{\psi}.
\end{align}

 Let us introduce $(2(s+1)n-2k)$ syndrome qudits in state $\ket{0}$, then, we obtain
 \begin{align}
 E\ket{\psi}\ket{0}^{\otimes (2(s+1)n-2k)} = \sum_{B\in\mathcal{P}^{\otimes (s+1)n}} a_B B\ket{\psi}\ket{0}^{\otimes (2(s+1)n-2k)}.
 \end{align}
 For the basis error $B$, the syndrome $\ket{s_B}$ is obtained based on the eigenvalues of the stabilizers with respect to $B\ket{\psi}$. Let $\mathscr{S}$ be the syndrome computation operator that transforms $B\ket{\psi}\ket{0}^{\otimes (2(s+1)n-2k)}$ to $B\ket{\psi}\ket{s_B}$. Then we operate $\mathscr{S}$ on $E\ket{\psi}\ket{0}^{\otimes (2(s+1)n-2k)}$, and  obtain
 \begin{align}
\mathscr{S} E\ket{\psi}\ket{0}^{\otimes (2(s+1)n-2k)} &= \mathscr{S}\left(\sum_{B\in\mathcal{P}^{\otimes (s+1)n}} a_B B\ket{\psi}\ket{0}^{\otimes (2(s+1)n-2k)}\right)\nonumber\\&= \sum_{B\in\mathcal{P}^{\otimes (s+1)n}} a_B \mathscr{S}\left(B\ket{\psi}\ket{0}^{\otimes (2(s+1)n-2k)}\right),\nonumber\\
&= \sum_{B\in\mathcal{P}^{\otimes (s+1)n}} a_B B\ket{\psi}\ket{s_B}.\label{eqn:SyndComp_mid1_bit}
 \end{align}
As $\ket{s_B}$s are of the form $\ket{s_1}\ket{s_2}\dots\ket{s_{(2(s+1)n-2k)}}$, they are orthogonal states for correctable errors. Thus, on measuring these syndrome qudits, the measurement outcome is $s_B = [s_1~s_2~\dots\\s_{(2(s+1)n-2k)}]$ for some $B$ with the post-measurement state being $B\ket{\psi}\ket{s_B}$. Also, using the syndrome $s_B$, the error is deduced, and the inverse error $B^{\dagger}$ is applied. Here, the syndrome qudits are discarded.

Alternatively, using control-based operations with the $(2(s+1)n-2k)$ syndrome qudits are control qudits and the codeword qudits as target qudits, the inverse error operator $B^{\dagger}$ is applied when the syndrome state is $\ket{s_B}$. Thus, the errors that are not Pauli basis errors are also corrected. We conclude that if we can correct a subset of Pauli basis errors, then we can correct errors that can be expressed as a linear combination of these errors.

\section{Conclusion}

In this paper, we have constructed many quantum codes over a class of finite commutative non-chain rings $\mathscr{R}_{q,s}$, with better parameters than the codes available in recent literature. Particularly, we have obtained $(\gamma,\Delta)$-cyclic codes using a set of idempotents over $\mathscr{R}_{q,s}$ and established results on their algebraic structure. Towards the construction of quantum codes, a necessary and sufficient condition to contain their dual codes has been established. Finally, we have obtained many better quantum codes. We have concluded our work by discussing the encoding and error correction capacity of our proposed quantum codes. However, exploring applications in the quantum computations of these codes is still open as future research work.

\section*{Acknowledgements}
The first and second authors are thankful to the Department of Science and Technology (DST), Govt. of India, for financial support under CRG/2020/005927, vide Diary No. SERB/F/6780/ 2020-2021 dated 31 December 2020 and Ref No. DST/INSPIRE/03/2016/001445, respectively.
\section*{Declarations}
\textbf{Data Availability Statement}: The authors declare that [the/all other] data supporting the findings of this study are available within the article. Any clarification may be requested from the corresponding author provided it is essential. \\
\textbf{Competing interests}: The authors declare that there is no conflict of interest regarding the publication of this manuscript.\\
\textbf{Use of AI tools declaration}
The authors declare that they have not used Artificial Intelligence (AI) tools in the creation of this manuscript.

	\end{document}